\documentclass[acmtog]{acmart}

\usepackage{booktabs} 
\newtheorem{problem}{Problem}[section]
\setcopyright{acmlicensed}
\acmJournal{TOG}
\acmYear{2023}
\acmVolume{42}
\acmNumber{6}
\acmArticle{231}
\acmMonth{12}
\acmPrice{15.00}
\acmDOI{10.1145/3618394}

\usepackage{bm}
\usepackage{siunitx}
\usepackage{wrapfig}
\usepackage{amsmath}
\usepackage{amsfonts}
\usepackage{amsthm}
\usepackage[ruled,vlined,noend,linesnumbered]{algorithm2e}
\usepackage{setspace}
\usepackage{graphicx}
\usepackage{overpic}
\usepackage[]{mdframed}

\newtheorem{prop}{Proposition}

\definecolor{RCColor}{rgb}{0.15, 0.38, 0.68}
\newcommand{\RC}[1]{}
\definecolor{DZColor}{rgb}{0.15, 0.68, 0.38}
\newcommand{\DZ}[1]{}

\def\tr{\mathop{\rm tr}}

\newcommand{\bR}{\mathbb{R}}

\newcommand{\cE}{\mathcal{E}}

\newcommand{\cP}{\mathcal{P}}

\newcommand{\tl}{{\tilde{\lambda}}}
\newcommand{\tM}{{\tilde{M}}}
\newcommand{\scp}[1]{\scshape{#1}}
\newcommand{\hTh}{\hat\Theta}

\newcommand{\del}{\mathrm{Del}}

\newcommand{\Dl}{\Delta \lambda}
\newcommand{\DF}{\nabla F}

\newcommand{\shear}{\sigma}
\newcommand{\asum}{\Sigma}

\begin{document}

\title{Metric Optimization in Penner Coordinates}

\author{Ryan Capouellez}
\affiliation{%
  \institution{New York University}
  \country{USA}
}
\email{rjc8237@nyu.edu}

\author{Denis Zorin}
\affiliation{%
  \institution{New York University}
    \country{USA}
}
\email{dzorin@cs.nyu.edu}

\begin{abstract}
Many parametrization and mapping-related problems in geometry processing can be viewed as metric optimization problems, i.e., computing a metric minimizing a functional and satisfying a set of constraints, such as flatness. 

\emph{Penner coordinates} are global coordinates on the space of metrics on meshes with a fixed vertex set and topology, but varying connectivity, making it homeomorphic to the Euclidean space of dimension equal to the number of edges in the mesh, without any additional constraints imposed. These coordinates play an important role in the theory of discrete conformal maps, enabling recent development of highly robust algorithms with convergence and solution existence guarantees for computing such maps. 

We demonstrate how Penner coordinates can be used to solve a general class of optimization problems involving metrics, including optimization and interpolation, while retaining the key solution existence guarantees available for discrete conformal maps.

\end{abstract}

%
%

%
%


\begin{teaserfigure}
    \centering
   \includegraphics[width=\textwidth]{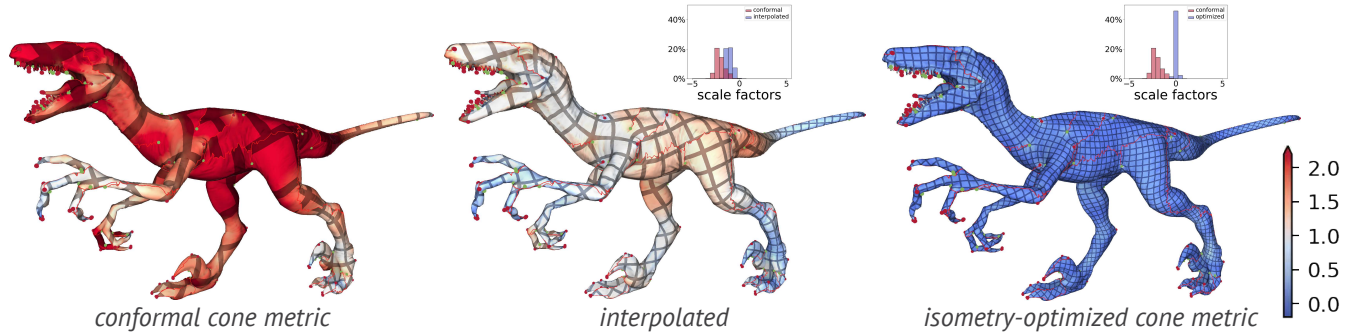}
    \caption{
    Examples of parametrizations produced using our method, interpolating in metric Penner coordinates between  conformal and isometry-optimized maps. Cones are marked by red and green points, and red paths indicate cuts connecting cones and cutting the surface to a disk. A grid texture on the parametrization is mapped back to the surface. The best fit scale factors (Section \ref{sec:objectives}) are also visualized as shading on the surface; these scale factors measure the local area distortion of the parametrization.
    }
    \label{fig:teaser}
\end{teaserfigure}

\maketitle

\section{Introduction}
\label{sec:intro}
A number of common geometry processing operations can be viewed as optimizing the metric of a surface with respect to a quality measure and subject to a set of constraints.  The most common  problems in this class are related to surface parameterization. Surface parameterization is usually defined as computing a map, or a collection of maps, from the surface, cut to one or more topological disks, to the plane.   Alternatively, these maps can be viewed as an almost everywhere flat metric on the surface (i.e., a cone metric, with nonzero curvature concentrated at a small number of vertices) minimizing a distortion measure.  This approach makes it possible to use intrinsic variables, does not require an arbitrary of choice of cuts, and turns out to be more natural for settings, e.g., requiring constraints on angles, common for global parametrization problems.  

A natural way to represent a cone metric on a triangle mesh is by assigning lengths to the edges. This representation highlights the challenge of the problem, whether we view it as a metric construction or a mapping problem: for a fixed connectivity, not every cone metric can be represented by edge lengths as the range of representable metrics is constrained by triangle inequalities on each face. It is also in general unclear, for a fixed connectivity, if there is an edge length assignment that satisfies all required constraints (most importantly, flatness/prescribed total angles at vertices). In other words, \emph{it is not known, for many important types of parametrization problems, if a solution exists for any specific fixed mesh connectivity.}  

While such problems remain open for fixed mesh connectivities, many existence and uniqueness questions have been answered \emph{for conformal maps} when the problem is generalized to the whole space of flat cone metrics with a given vertex set and surface topology \cite{gu2018discrete, springborn2019ideal}. These recent theoretical breakthroughs have yielded practical algorithms to compute arbitrarily prescribed cone metrics that are (discretely) conformally equivalent to the metric of the original surface \cite{gillespie2021discrete, campen2021efficient}.  Importantly, this requires considering connectivity changes.  

An essential component of this theory is the definition of coordinates on the space of metrics.  For any flat metric with given vertices, there are connectivities in which it can be described by edge lengths. At the same time, the same metric can be described by different connectivities: e.g., an \emph{intrinsic edge flip}  produces a different connectivity with a different length assignment representing the same metric. That is, the space of cone metric is covered by the overlapping subsets of metrics defined by edge lengths on given connectivities.

To be able to navigate the complete space of flat metrics, we need to consider connectivity changes, which seemingly necessitates frequent discontinuous changes in metric representation.  Another major difficulty in using variable connectivity in the context of optimization problems is that all known distortion measures are \emph{not} invariant with respect to common connectivity changes (e.g., edge flips).

Fortunately, it turns out that the whole space of cone metrics can be parametrized with quantities defined on \emph{fixed} connectivity known as \emph{Penner coordinates}, which we describe in more detail in Section~\ref{sec:penner}. 
These coordinates establish a \emph{bijection} between the space of metrics with a fixed set of cone vertices and $\bR^{|E|}$,  where $|E|$ is the (fixed) number of edges in triangulations with a fixed number of vertices and surface genus. 
Penner coordinates are not restricted by triangle inequalities, allowing for unconstrained optimization in the space of metrics and ensuring the existence of a solution in many cases (e.g., for convex objectives). 

In this paper, we demonstrate how Penner coordinates can be used computationally to perform metric optimization and interpolation  not restricted to the class of conformal maps. 

We demonstrate that for a class of objectives, one can design optimization methods working in Penner coordinate space such that 
\begin{itemize}
\item  a choice of metric satisfying constraints (e.g., flatness)  at vertices is \emph{guaranteed} to exist, possibly requiring remeshing;  
\item  there are natural distortion measures, related to commonly used ones, that can be extended to the whole space of metrics when expressed in Penner coordinates; 
\item  For suitable objectives, the optimal metric can be computed using reliable numerical methods. We present two formulations: an unconstrained optimization problem in shear coordinates, obtained by linear transformation of Penner coordinates,  for which any gradient-based method can be used, and a more efficient coordinate-projection algorithm treating constraints in an  explicit way.  
\end{itemize}

We test our method on the complete dataset \cite{Myles:2014}, 
and its variations used in \cite{campen2021efficient}, which contains many challenging examples, and on a version of this dataset cut to disks and with boundary angle prescriptions that force extremely high distortion.  We demonstrate that many orders of magnitude of reduction of the distortion can be obtained while still satisfying the constraints exactly.

\section{Related Work}
\label{sec:related}

There is a wealth of work on various types of metric optimization problems, most of it focused 
on parametrization, i.e., computing (almost) everywhere flat metrics, with a variety of 
constraints.   These methods are distinguished by the choice of variables, optimization objectives, 
type of constraints considered, guarantees these methods provide and initialization requirements.
We briefly review most closely related work, and refer to recent surveys 
 \cite{naitsat2021inversion} and \cite{fu2021inversion} for a more comprehensive review. 

The most fundamental choice is the choice of metric representation, with the main distinction 
between  \emph{intrinsic} variable methods, using e.g., lengths, angles or conformal scale factors, 
and parametric coordinate methods (restricted to parametrization) with planar vertex coordinates  as variables. While these methods share a number of common challenges, e.g., for most methods, 
flat metric/injectivity constraints are nonlinear, these typically address distinct categories 
of problems, as they differ in their classes of natural constraints: parametric 
positional constraints are easy to formulate in $u,v$ representation, while constraints on lengths 
or angles are nonlinear and hard to impose, whereas for intrinsic variables the situation is 
reversed. 

\paragraph{Intrinsic methods.} Most significant earlier works in this category include 
\cite{sheffer2001parameterization, kharevych2006discrete,springborn2008conformal,BenChen:2008}, 
these works focus on conformal maps, using angles, logarithmic radii and scale factors 
as intrinsic variables.  In \cite{kharevych2006discrete}, intrinsic Delaunay triangulations 
\cite{bobenko2007discrete} are used, similar to more recent methods for conformal maps 
\cite{sun2015discrete}, \cite{campen2021efficient} and \cite{gillespie2021discrete}. These methods, the first  providing full guarantees (up to the errors introduced by finite numerical precision) on 
discrete conformal map construction, build on the mathematical foundations  developed in 
\cite{gu2018discrete,gu2018discrete2,springborn2019ideal} with key concepts 
originating in \cite{penner1987decorated,rivin1994euclidean}.  Our approach further develops
this approach applying it to metric deformations that are not necessarily conformal, 
while retaining many of the desirable features. In contrast with the per-vertex logarithmic scale factors used in prior work, we use per-edge \emph{logarithmic length} coordinates that cover the entire space of cone metrics

Conformal map methods naturally support constraints on angles at vertices, or more generally linear combinations of angles along loops,  and boundary lengths (through fixing scale factors).  These features are used in \cite{Campen:2017:SimilarityMaps} and \cite{campen2018seamless}) to construct 
\emph{similarity} maps and \emph{seamless} maps for global surface parametrization. 

One important feature of these methods is that the problem is reduced (with some caveats in earlier methods) to convex optimization. As a consequence these techniques, although nonlinear, 
do \emph{not} require initialization, or, rather any initialization can be used successfully, 
even if it does not satisfy the constraints.  The versions of conformal parametrization 
methods that guarantee solution existence do not preserve connectivity; two main 
approaches to connectivity changes were proposed: the one we follow is based 
on maintaining Delaunay property of the meshes; the alternative 
\cite{luo2004combinatorial,springborn2008conformal} is to perform surgery when triangles 
degenerate; this is less desirable as the transition between different connectivities 
involves degeneracies.

\paragraph{Parametric coordinate methods.} Methods using 2D coordinates in the parametric 
plane form a different category. Compared to intrinsic methods, these naturally support 
a different class of constraints, e.g., positional constraints, while angle and length 
constraints are more problematic in this setting.  In this class, there are several methods that provide 
local or global bijectivity guarantees, without requiring an initial map, most importantly, 
\cite{Tutte}, on which most other methods with guarantees are based:
\cite{weber2014locally}, \cite{ProgEmbedding}.  While intrinsic methods naturally 
work with arbitrary closed surfaces and through doubling with surfaces with boundary \cite{sun2015discrete}; 
parametric coordinate methods require cutting meshes to disks, and defining target 
boundaries (simple, or \emph{self-overlapping} polygons), which may not be easy to construct 
e.g. for seamless parametrization problems (cf. \cite{Zhou:2020}, \cite{levi2021seamless}).
These methods do not naturally support free boundaries, although can be used as 
starting points of free boundary methods. 

We briefly mention a number of techniques that assume an initial parametrization (usually 
it comes from  Tutte's map, although intrinsic or methods like \cite{weber2014locally} can be used) and then deform the parametrization or metric while  maintaining bijectivity, using various types of barrier energies \cite{Schueller:LIM:2013}, \cite{Rabinovich:2017:SLI}, \cite{liu2018progressive}.  These methods require a feasible starting point, so need to be augmented by a different method that yields one.
Constraint-convexification approaches \cite{Lipman:2012} in principle, can be started from unfeasible points, but may have no 
feasible solutions in some cases when the original problem has one. 

As initializing parametric coordinate methods with a feasible solution is often difficult, a number of techniques were proposed to avoid the need for this.  This limits the space of available solutions, so if a feasible solution is not found, it does not mean it does not exist. A  number of recent promising methods aim to produce high-quality locally injective maps without a feasible starting point, or connectivity changes.  For example, \cite{du2020lifting}  introduces a novel total lifted content energy which  has the property that its \emph{global} minimum is an injective embedding, if one exists. \cite{overby2021globally} proposes another method in this category. It demonstrates an impressive practical success rate for meshes with fixed connectivity and fixed boundary, but the existence of a solution or an attainable local minimum is not guaranteed.  
In comparison, our focus is on demonstrating that an important range of problems are guaranteed to have a solution in the space of metric, although connectivity changes may be required.

\section{Problem Formulation and Overview}
\label{sec:problem}

We start by discussing the basic problem that motivates our work. 

A \emph{discrete metric} for a triangular mesh $M = (V,F,E)$ is defined by an assignment of lengths $\ell : E \rightarrow \bR^+$, satisfying triangle inequality. The discrete metric naturally defines a \emph{cone metric} on the mesh, whose restriction to each triangle is just the planar metric. All curvature of this metric is concentrated at mesh vertices.

\begin{problem}
 For a given discrete metric $(M,\ell^0)$,  let  $\Theta_i$ be the sum of angles of triangles sharing a vertex $i$,   Given \emph{target angles} $\hat\Theta_i$ (respecting the discrete Gauss-Bonnet theorem), compute a new metric $(M,\ell)$ for which the sums of angles at vertices have values  $\hat\Theta_i$ while minimizing an objective $E(\ell^0,\ell)$.
\label{prob:problem-fixed}
\end{problem}
 This objective $E$ is  a measure of distortion between the initial metric $\ell^0$ and current one $\ell$.  If the mesh is a disk, and $\hat\Theta_i = 2\pi$ for all interior $i$, this reduces to a disk parametrization problem with prescribed angles on the boundary. 
 
Unfortunately, this problem may not have a solution for a fixed connectivity $M$ because  the optimization is done within the domain in $\bR^{|E|}$ defined by the strict positivity of edge lengths and triangle inequality constraints, i.e., the feasible domain does not include its boundary.   The constrained optimization problem may have a solution on the boundary of this domain, i.e., containing infeasible degenerate triangle configurations, suggesting that connectivity changes are necessary to find an optimal discrete metric with non-degenerate triangles. 

The approach we explore extends the domain of Problem~\ref{prob:problem-fixed} from the space of discrete metrics $(M, \ell)$ on a fixed connectivity $M$ to the space of \emph{all} cone metrics for a given topology of genus $g$ with a fixed set of vertices. 

\begin{problem}
Given \emph{target angles} $\hat\Theta_i$ (respecting the discrete Gauss-Bonnet theorem), compute a new cone metric $(M',\ell')$ with new triangle-mesh connectivity $M'$ with the same vertices and genus as $M$, for which the sums of angles  at vertices have values  $\hat\Theta_i$, while minimizing an objective $E( M, \ell^0, M',\ell')$.
\label{prob:problem-variable}
\end{problem}

\newcommand{\ve}[1]{#1}
\begin{algorithm}[b]
\setstretch{0.9}
\SetAlgoLined
\DontPrintSemicolon
\SetKwInOut{Input}{Input}
\SetKwInOut{Output}{Output}
\SetKwProg{Fn}{Function}{:}{}
\SetKwRepeat{Do}{do}{while}
\SetKw{Not}{not}
\Input{
    triangle mesh $M = (V,E,F)$, closed, manifold,\newline
    edge lengths $\ell = e^{\lambda/2} > 0$ satisfying triangle inequality,\newline
    target angles $\hat\Theta > 0$ respecting Gauss-Bonnet, a per edge distortion measure $E(\lambda)$
}
    \vspace{2pt}
\Fn{\scp{Optimize}$(M,\lambda,\hTh)$}{
Decompose $\lambda = Sx + Bu$\;
\For{ $i = 0, 1, ...$}{
    $\hat{\lambda} \gets S x$ \;
    $\lambda, e_1,...,e_n \gets$ \scp{FindConformalMetric}$(M, \hat{\lambda}^i, \hTh)$ \;
    $\tM_0 , \tl_0 \gets M, \lambda$ \;
    \For{ $j = 1,...,n$ }{
        $\tM_j, \tl_j, D_j \gets$ \scp{DiffPtolemyFlip}$(\tM_{j - 1}, \tl_{j - 1}, e_j)$ \;
    }
    $\tM, \tl \gets \tM_n, \tl_n$ \;
    $\alpha, \nabla_{\tl}\alpha \gets$ \scp{ComputeAnglesAndGradient}$(\tM,\tl)$ \;
    $\nabla_{\lambda} F \gets \Sigma \, \nabla_{\tl} \alpha \, \prod_{j=1}^n D_j$ \; 
    $\nabla_{x} E_S \gets \nabla_{\lambda} E (S - B(\nabla_{\lambda} F \, S)^{-1}(\nabla_{\lambda} F \, B))$ \;
    $x  \gets x -\beta \nabla_{x} E_S$\; 
    }
}
\caption{Angle constraint space coordinate gradient descent summary.}
\label{alg:grad-descent}
\end{algorithm}

The algorithms we develop for solving this problem also yield a map from the modified mesh $PL(M',\ell')$ to the original mesh $PL(M,\ell^0)$, where $PL(M,\ell)$ denotes the piecewise linear (PL) mesh associated with the combinatorial mesh $M$, with triangle edge lengths $\ell$.  Using this map, we can produce a refinement $PL(M',\ell')^r$ satisfying the input angle constraints at all original vertices of $M$, and flat at all inserted vertices.  Equivalently, this yields a $uv$ map from $PL(M,\ell^0)^r$ to the plane,  i.e., solve the standard parametrization problem with guarantees for angle constraints by allowing refinement. 

To solve this new  problem computationally using standard gradient-based methods,  we rely on four main ingredients: 
\begin{itemize}
\item Coordinates on the space of all cone metrics with a given topology and vertex set (Section~\ref{sec:penner}) that allow expressing suitable objectives and constraints.
\item Definition of the constraint space $\Theta_i(\ell) = \hat\Theta_i$, and  coordinates for the constraint space (Section~\ref{sec:projection}).
\item Objectives measuring distortion that are defined on the whole space of metrics (Section~\ref{sec:objectives}).
\item Computation of gradients of the objectives  with respect to the constraint-space coordinates (Section~\ref{sec:gradients}).
\end{itemize}

While these four components make the problem amenable to all standard gradient optimization methods, and the algorithms inherit the usual guarantees, e.g., for gradient descent or BFGS, we found that a substantial  acceleration can be achieved by using a coordinate-projection gradient descent
(Section~\ref{sec:projected}).

As a preview, we summarize the simplest form, gradient descent optimization, as Algorithm~\ref{alg:grad-descent}.

First, the input logarithmic edge lengths are decomposed into shear and scale components $x$ and $u$ by a linear transformation (line 2, Section~\ref{sec:projection}).

On every iteration of gradient descent, the gradient of the shear distortion measure $E_S$, which is the restriction of $E$ to metrics parameterized by the shear variables 
$x$ that serve as the free variables in the algorithm, is computed. 

The computation involves solving for a 
conformally equivalent metric $\tilde{\lambda}$ for the lengths inferred from the shears (line 5). 
The conformal mapping function  \cite{campen2021efficient} produces a new connectivity $\tilde{M}$, and a sequence of flips $e_j$, $j=0 \ldots n$, that lead to it.

For each flip, the function {\sc DiffPtolemyFlip} (line 8, Section~\ref{sec:penner} and Section~\ref{sec:gradients}) computes new lengths using the Ptolemy formula, and concurrently computes matrices $D_j$ corresponding to the derivatives of these transformations of lengths.

Angles and their gradients with respect to logarithmic lengths are computed in the new connectivity $\tilde{M}$ (line 10).

Finally, the gradient of the distortion $E_S$ is computed from the angle gradients $\nabla_{\tl}\alpha$, matrices $D_j$,
angle summation matrix $\Sigma$, and shear/scale decomposition matrices $B$ and $S$ (lines 11-12).  All these quantities are defined more precisely in subsequent sections. 

Finally the shear variable $x$ is updated (line 13). 

In the next sections, we describe each component separately.

\section{Penner coordinates for cone metrics} 
\label{sec:penner}

We first describe \emph{Penner coordinates} that we use as coordinates on the space of all flat cone metrics with a fixed topology of genus $g$ and set of cone vertices $V$. We denote this space by $\mathcal{C}_{g, V}$.  We start by defining an (almost) unique choice of triangulation for any metric in this space; this leads to partitioning of the space of metrics into Penner cells, each corresponding to a distinct choice of triangulation.  Then we show how discrete metric coordinates (edge lengths) on an arbitrary chosen cell can be translated to any other cell, leading to identification of the space of metrics with Euclidean space. 

\paragraph{Intrinsic Delaunay triangulations.}
For a given cone metric induced by a discrete metric $(M,\ell)$, the choice of triangulation  $M$ is  not unique, because, e.g., an intrinsic edge flip,  (Figure~\ref{fig:intrinsic}) does not change the cone metric.  
\begin{figure}
    \centering
    \includegraphics{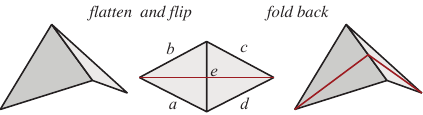}
    \caption{Intrinsic edge flip: two adjacent triangles are unfolded to the plane, and then a standard flip is performed; the flipped edge corresponds to a broken line on the original geometry.}
    \label{fig:intrinsic}
\end{figure}
However, for  any cone metric, there is an almost unique canonical choice of triangulation, specifically, the \emph{intrinsic Delaunay triangulation} \cite{bobenko2007discrete,fisher2007algorithm}. 

On a surface with cone metric,  a triangulation can be defined as an embedded graph with edges corresponding to non-intersecting geodesics connecting vertices (cones), partitioning the surface into triangular domains. 
As the metric is flat away from vertices, each such domain is isometric to a triangle. Then the intrinsic Delaunay condition for an intrinsic edge $e$ is defined as the standard condition $\alpha + \beta \leq \pi$  on triangle angles $\alpha$ and $\beta$ opposite $e$ in two incident triangles.  

\paragraph{Penner cells.}
For a fixed connectivity, $\mathcal{C}_{g, V}$ has natural coordinates, the edge lengths, as long as the resulting triangulation stays Delaunay. 
\begin{definition}
For a fixed connectivity $M = (V,E,F)$,  we  consider all possible lengths assignments to the edges that define cone metrics for which mesh $M$ is Delaunay. The set of cone metrics defined by these lengths assignments is called the \emph{Penner cell} $\cP(M) \subset \mathcal{C}_{g, V}$.
The coordinates for metrics contained in this set are given by the $|E|$ lengths on the edges of $M$.
\end{definition}
Note that Penner cells are closed. Moreover, a hyperface
$F(T_1, T_2)$ of the  boundary of a Penner cell consists of all metrics such that, for a pair of adjacent triangles $(T_1, T_2)$ in $M$, their four vertices are co-circular with respect to the metrics. 
Observe that the  connectivity $M'$ obtained by flipping the common edge of $T_1$ and $T_2$ is also Delaunay for the same metric, i.e., $F(T_1, T_2)$ is shared by  Penner cells $\cP(M)$ and $\cP(M')$. 

The collection of Penner cells covers the whole space $\mathcal{C}_{g, V}$ of cone metrics \cite{springborn2019ideal}.
Unlike the set of discrete metrics for fixed connectivity, this set, as we show below, can be identified with $\bR^{|E|}$. 

\begin{definition}
If $\ell$ is a choice of lengths for the mesh $M$,
for which $M$ is not necessarily Delaunay, we 
define a \emph{function $\del$}, mapping $M$ to a different connectivity $M'$
with lengths $\ell'$ with the same vertices and number of edges,  and  such that $M'$ is Delaunay with respect to $\ell'$ and defines the same metric.
\[
\del(M,\ell) = (M',\ell')
\] 
\end{definition}
$\del(M,\ell)$ can be obtained from $(M,\ell)$ by the standard flip algorithm for Delaunay triangulation. 
We also denote the Delaunay connectivity $M'$ for the metric $(M,\ell)$ by $M_{\del}(M,\ell)$.  

This cell partitioning of the space of metrics, with individual edge length coordinates on each cell, already allows for unconstrained optimization in the space of metrics using the standard optimization-on-manifold approaches. For example,  for a line search, typically used in gradient-based optimization method in a direction $d$ at a point inside a cell, one can move along $d$ in the local  coordinates in the cell  $\cP(M)$  until we reach the cell boundary separating it from a connectivity $\cP(M')$. Then we change triangulation and length coordinates along the line, to that of an adjacent cell, using an intrinsic flip. 
This however is undesirable for more complex situations, as the step of optimization is limited by the distance 
to the nearest Penner cell boundary, which can be very small. Even more significantly, distortion measures that change continuously, let alone smoothly, with respect to edge flips are rare:  the energy used to compute conformal maps is unusual in this respect.  
To be able to navigate over the whole space of metrics and extend distortion measure definitions to  all of  $\mathcal{C}_{g, V}$, we introduce \emph{global} Penner coordinates on the whole space of metrics. 

\paragraph{Ptolemy transition maps.}
The transition between two charts for a metric contained in the hyperface separating two Penner cells, amounts to  applying the \emph{Ptolemy formula}, allowing us to compute the
length of the flipped edge if the pair of triangles incident at the edge have co-circular vertices. 
Removing $e$ and inserting  the flipped edge $e'$ in a pair of adjacent triangles with external edges  $\ell_a,\ell_b,\ell_c,\ell_d$  corresponds to the edge length update
\[
\ell'(e') = \frac{\ell(a)\ell(c) + \ell(b)\ell(d)}{\ell(e)}, 
\]
and  $\ell'(f)  = \ell(f)$ for all edges $f \neq e$.

This \emph{transition map} $\tau(M,M'): \bR^{|E|} \rightarrow \bR^{|E|}$ between length coordinates for cells $\cP(M)$ and $\cP(M')$ is a critical  element of our construction. 

We make use of the following observations, summarized in a proposition.

\begin{prop}
\begin{enumerate}
\item The transition maps $\tau(M,M')$ are  \emph{smooth} (in fact analytic), as a consequence, the atlas formed by these coordinate charts on the space of metrics is $C^\infty$; 
\item The transformations between coordinates on non-adjacent cells,  connected by a sequence of flips of edges $e_1,e_2,\ldots e_n$, is given by the  composition $\tau_n\circ\tau_{n-1}\ldots \tau_1$.  
This change of coordinates can be done from any cell to the cell $\cP(M)$ that corresponds to the initial metric  $(M,\ell^0)$ \emph{however,  the resulting values $\ell(e)$, obtained by a sequence of Ptolemy flips,  need not satisfy triangle inequalities}
\item  $\tau(M, M')$ does not depend on the sequence of cells used to construct the map. 
\end{enumerate}
\label{prop:summ}
\end{prop}
Except for the last statement, these observations are  direct consequences of the definition of $\tau$ as 
a simple rational function on lengths.  The last  statement directly follows from the fact that Ptolemy
flips in fact preserve an ideal hyperbolic metric on the  surface \cite{penner1987decorated}.

\paragraph{Penner coordinates.}
If we fix an arbitrary connectivity $M_0$, we can extend the discrete length coordinates on this cell to coordinates for the whole space of metrics. We refer to this extension as Penner coordinates with respect to $M_0$.

\noindent

\begin{mdframed}
\begin{definition}
Penner coordinates for a cone metric with length coordinates $(M,\ell)$ in Penner cell $\cP(M)$, with respect to
$M_0$ is a vector $P_{M_0}(M, \ell)$ of positive numbers in $\bR^{|E|,+}$ defined as  
\[
P_{M_0}(M, \ell)  = \tau(M,M_0)(\ell).\] 
i.e., simply the coordinate change from $\cP(M)$ to $\cP(M_0)$ by a composition of Ptolemy formulas.  
\end{definition}
\end{mdframed}

\begin{figure}
    \centering
    \includegraphics[width=2.5in]{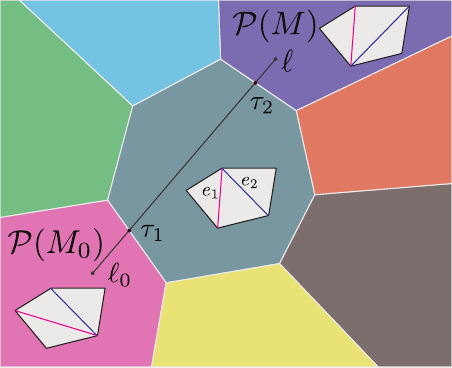}
    \caption{A schematic illustration of transitions between different Penner cells.}
    \label{fig:transition_map}
\end{figure}

Following discrete conformal map theory,  we use logarithms  $\lambda = 2\ln \ell$ rather than Penner coordinates themselves to eliminate the  positivity constraints; the factor 2 is introduced to simplify some expressions.

We  emphasize that $P_{M_0}(M, \ell)$, while formally obtained using the Ptolemy formula, are \emph{not} 
lengths: unless $M=M_0$, these are not guaranteed to satisfy triangle inequality.  Ptolemy formula yields Euclidean lengths only if the flip is performed on a pair of triangles with co-circular vertices. 

\begin{prop}
Logarithmic Penner coordinates $P_{M_0}(M,\ell)$ define a
bijection between the space of cone metrics and the Euclidean space $\bR^{|E|}$.
\end{prop}
\begin{proof}
As this is a key fact we use in this paper, we summarize a proof for completeness.
It is  well-known that any two triangle meshes with the same topology are related by a sequence of flips. 
Then for any metric,  Penner coordinates are defined, and  uniquely by independence 
of the coordinates obtained from the sequence of flips 
(Proposition~\ref{prop:summ}) from $M$ to $M_0$.
Thus the map from the metrics to Penner coordinates w.r.t. $M_0$
is well-defined. As Ptolemy flips are invertible, 
the map from metrics to $\bR^{|E|}$ defined by Penner coordinates is also injective. 

Conversely, if we prescribe any set of coordinates $\lambda$
on a given connectivity $M_0$, Weeks' algorithm \cite{weeks1993convex} discussed below, can be used
 to find a connectivity $M'$ and assignments $\lambda'$
 for which the metric defined by $\exp(\lambda'_e / 2)$
 for all edges $e$ is Delaunay, and in particular satisfies triangle inequality.  Thus the map from metrics to $\bR^{|E|}$
 is surjective, i.e., a bijection.    
\end{proof}

\paragraph{Example.} Figure~\ref{fig:penner_cells_three} shows a simple example 
of Penner cell partition of the space of metrics with three vertices and sphere topology. 
There are two possible topological triangulations of three points (there is no isometric embedding of these metric in 3D, but these are valid flat metrics that may arise, e.g, when parametrizing a sphere with two curved triangles). 

Although $|E|=3$ and the space of metrics in this case is 3-dimensional, due 
to scale invariance, the structure of the Penner partition can be shown in 2D  by restricting logarithmic edge length to sum up to zero.  The curves correspond to Delaunay inequalities  becoming equalities. 

\begin{figure}[htb!]
    \centering
    \includegraphics[width=3in]{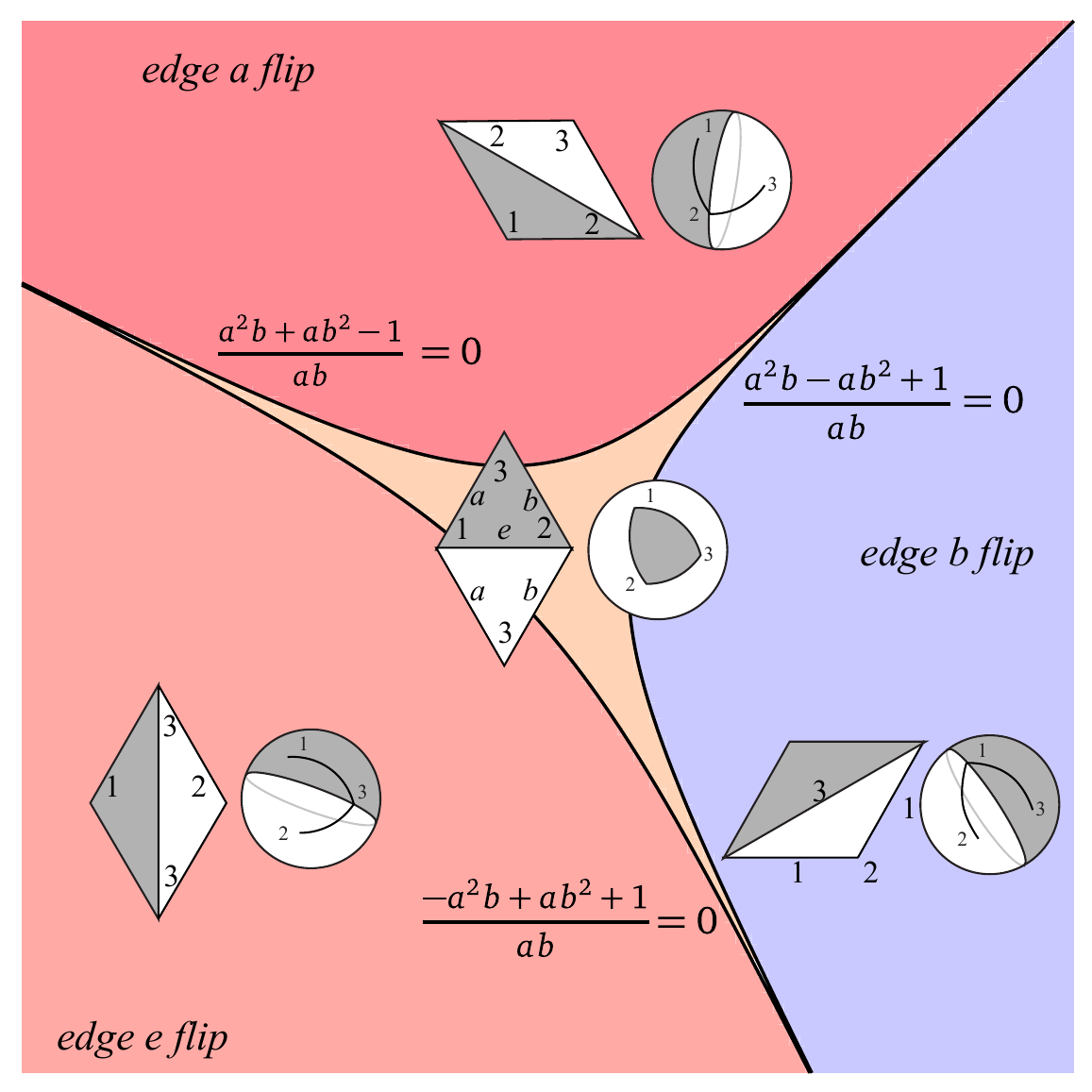}
    \caption{Example of a Penner cell partition of the space of cone metrics with 3 vertices and 3 edges. It has four cells, corresponding to mesh connectivities with vertex degrees (2,2,2), and 3 versions with degrees (1,1,4), corresponding to three possible flips.     The cells are shown in logarithmic coordinates $\ln a$, $\ln b$, in the plane with equation $\ln a + \ln b+ \ln e = 0$. 
    We show two triangles of each configuration laid out in the plane after cuts along two edges, and on a sphere as curved triangles, to illustrate the connectivity of the edge graph more explicitly. 
    }
    \label{fig:penner_cells_three}
\end{figure}

\paragraph{Computing Penner coordinates for a metric.}  
While Penner coordinates form a convenient parameterization of the whole space of cone metrics, in general,
we cannot compute standard geometric quantities directly from 
these coordinates, as the triangle inequality is not satisfied. However, we can do this if we use coordinates 
with respect to the Penner cell $\cP(M)$ containing the metric: in this cell, the mesh $(M,\ell)$ satisfies  Delaunay
conditions, and hence triangle inequality \cite{weeks1993convex}.  To perform this change of coordinates, we need to apply the inverse transition maps $\tau(M_0,M)(\ell^0)$, which requires knowing the Delaunay connectivity $M$, and  the sequence of flips leading from $M_0$ to $M$.  Neither may be known: e.g., if in the process of optimization 
we update our coordinates $\lambda + \Delta \lambda$, we have no way of knowing in which Penner cell we landed. 

However, remarkably, one can apply the usual  Delaunay triangulation algorithm to $\ell^0$, 
despite the fact that triangle inequality is not satisfied; in this form, the algorithm is known as \emph{Week's algorithm} \cite{weeks1993convex}. This fact follows from the equivalence of  ideal decorated hyperbolic metrics and Euclidean metrics on triangles  discussed in detail in numerous papers; for a brief summary,  please see \cite{gillespie2021discrete} and \cite{campen2018seamless}. 

The "ideal hyperbolic" Delaunay condition for two triangles as in Figure~\ref{fig:intrinsic} has exactly the same form as the 
standard Delaunay written in terms of cosines: 
$\cos(\alpha) + \cos(\beta) \geq 0$, where $\alpha$ and $\beta$ are 
two angles opposite an edge,  but with cosines replaced with 
cosine law formulas, which do not require triangle inequality to 
be satisfied.
\begin{definition}
The edges of two triangles with edges $e,a,b$ and $e,c,d$ respectively, sharing an edge $e$, with positive numbers 
$\ell(a)$,$\ldots$,$\ell(e)$ assigned to edges, satisfy the ideal 
hyperbolic Delaunay condition if
\begin{equation}
\frac{\ell(a)^2 + \ell(b)^2-\ell(e)^2}{2\ell(a)\ell(b)} + 
\frac{\ell(c)^2 + \ell(d)^2-\ell(e)^2}{2\ell(c)\ell(d)} \geq 0 
\label{eq:ideal-delaunay}
\end{equation}
\end{definition}
The standard flip algorithm (repeatedly flip any edge that does not satisfy the condition, until none are left, using the Ptolemy formula)  in this case is called \emph{Week's algorithm} \cite{weeks1993convex}. It has the following property:

\begin{prop}
For any starting Penner coordinates, Week's algorithm produces 
a triangulation  $M$ satisfying the Delaunay condition; moreover, 
the resulting  $(M, \ell)$ satisfy triangle inequalities.
\end{prop}

In other words, it produces exactly the sequence of flips 
leading to the Penner cell corresponding to the metric with 
Penner coordinates $\ell$, and allows us to define the map 
$\del$ to arbitrary Penner coordinates, with the following property
 \[
\del(M_0, P_{M_0}(M,\ell)) = (M,\ell) 
\]

In summary, we defined a global coordinate system on the space of 
cone metrics with a vertex set $V$ which have fixed topology
(it follows from this that the number of edges is fixed), 
bijectively and continuously mapping it to $\bR^{|E|}$;  one can convert these coordinates to standard Euclidean lengths on a connectivity for which the mesh is Delaunay using Week's algorithm. 

\section{Angle constraint manifold and shear coordinates}
\label{sec:projection}

Fixing vertex angles imposes $|V|-1$ constraints on lengths, with one degree of freedom at an arbitrary vertex $v_0$ being redundant due to Gauss-Bonnet formula, and the dimension of the manifold of 
metrics with given vertex angles is $|E|-|V|+1$.
These constraints are nonlinear and nonconvex, so the manifold of metrics for any prescribed set of angles is not defined directly as a "nice" (e.g., convex) subset of $\bR^{|E|}$. 
However, we show that there is a \emph{linear} coordinate change on the logarithmic Penner coordinates 
that  defines a \emph{global} parametrization of the constraint manifold with $|E|-|V| + 1$ parameters (shear coordinates). 

\paragraph{Discrete conformal metrics.}  An essential component of our construction is the computation of a discrete conformal metric with prescribed angles. 

\begin{definition}
Two  cone metrics $(M,\ell)$ and $(M',\ell')$ with the same vertices are discretely conformally equivalent if 
there is a sequence of metrics 
$(M,\ell) = (M^0,\ell^0), \ldots (M^n,\ell^n)=(M',\ell')$ such that
\begin{itemize}
\item $M^m$ is Delaunay w.r.t., $(M^m,\ell^m)$;
\item $M^{m+1}=M^m$ and $\ell^{m+1}$ is obtained by scaling lengths $\ell^m$ 
as
\[\ell^{m+1}_{ij} = \ell^m_{ij}e^{(u^m_i + u^m_j)/2}\]
where $u^m_i$ are log scale factors assigned to vertices.
\item or $M^m$ and $M^{m+1}$ are two distinct Delaunay 
triangulations with the same cone metric. 
\end{itemize}
\end{definition}
Conformal changes of metric are parameterized by  logarithmic scale factors $u$,  i.e., have exactly one degree of freedom per vertex. As a consequence, the constraints on angles at vertices, as there are $|V|-1$ of these fully determine  the solution (up to a global scale factor).

In logarithmic variables, and fixed connectivity, the relation between lengths  and scale factors has a particularly convenient linear form 
\begin{equation}
\lambda_{ij} = \lambda_{ij}^0 + u_i + u_j
\label{eq:logconf}
\end{equation}
or, in matrix form, $\lambda = \lambda^0 + Bu$ for a matrix $B$ that only depends on the combinatorial structure of $M$. 

Moreover, this solution can be obtained by solving a convex minimization problem, \cite{springborn2008conformal}, 
which can be solved robustly, e.g., using the algorithms and software described in \cite{campen2018seamless}.
The convex energy $\cE(M,\ell,u,\hat\Theta)$ used to compute the scale factors $u$, by minimizing it with respect to $u$, is naturally formulated in terms of logarithmic lengths  $\lambda_{ij} = 2\ln \ell_{ij}$, and vertex log scale factors $u_i$. 

The condition for the minimum of $\cE$,  $\nabla_u \cE = 0$, is exactly the angle constraint, which we express as follows, in terms of Penner coordinates with respect to a reference connectivity $M_0$:

\begin{mdframed}
\begin{equation}
F(\lambda) =  \asum \alpha(\del(M_0,\lambda))-\hat{\Theta} = 0; \mbox{(Angle Constraint manifold)}
\label{eq:constrained_manifold}
\end{equation}
\end{mdframed}

where $\lambda = \lambda^0 + B u$, $\alpha(\lambda)$ is the vector of size $3|F|$ of angles of all triangles as functions of log lengths satisfying triangle inequalities, and $\asum$ is a $(|V|-1) \times {3|F|}$ matrix summing angles around each vertex (except one). 
We write the unique solution of Equation~\ref{eq:constrained_manifold}  with respect to $u$ as $u(\lambda^0)$, so that $F(\lambda^0 + Bu(\lambda^0)) = 0$, i.e., 
\begin{wrapfigure}{l}{0.25\linewidth}
\begin{overpic}[width=1.3\linewidth]{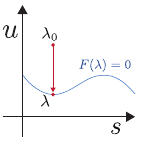}
\end{overpic}
\end{wrapfigure}
$u(\lambda^0)$ are the scale factors deforming the initial metric 
$\lambda^0$ so that the angle constraints are satisfied. 

\paragraph{Shear coordinates.}
The formulas above define a projection $\lambda^0 \rightarrow \lambda^0 + Bu(\lambda^0)$ to the constraint manifold.   Next, we show that a simple \emph{linear change} of logarithmic Penner coordinates defines a global parametrization of the constraint manifold, $\lambda_S(x)$, such that $F(\lambda_S(x))=0$ 
for any $x$. One can visualize the manifold  $F(\lambda) =0$ being a "graph" of a $(|V|-1)$-dimensional function (the vector of free scale factors) over a linear subspace of dimension $|E|-|V| + 1$.

Consider the column space $\mathcal{B} \subset \bR^{|E|}$ of matrix $B$, and the orthogonal complement of $\mathcal{B}$, which we denote $\mathcal{S}$; for reasons that will become clear, we refer to this latter space as the \emph{shear subspace}.  

Let $S$ be a full-rank matrix with columns spanning $\mathcal{S}$.  For now, we do not choose a specific basis for $\mathcal{S}$, possible choices are discussed below. Consider the change of coordinates 
$\lambda \rightarrow [u,x]$, defined by 

\[
\lambda = B u + S x = [B\; S] \left[\begin{array}{c}  u\\ x\\\end{array}\right]
\]
where $x$ are coefficients in the basis of columns of $S$.

As observed above, for any $\lambda = S x$, there is a unique $u(Sx) \in \bR^{|V| - 1}$ such that $F(S x + B u(Sx) ) = 0$ 
\vspace{0.5em}
\begin{mdframed}
\begin{equation}
\lambda_S(x) = Sx + B u(Sx) 
\label{eq:param-manifold}
\end{equation}    
\end{mdframed}

Note that for any $\lambda$ in the constraint manifold $\mathcal{M}_{\Theta} \subset \bR^{|E|} = \mathcal{S} \oplus \mathcal{B}$, we can apply the coordinate transformation $\lambda = S x+ B u$ to obtain $x$.

We conclude that the following proposition holds: 
\begin{prop}
    A metric with Penner coordinates $\lambda$  belongs to the angle constraint manifold 
    $F(\lambda) = 0$ if and only if  it has the form \eqref{eq:param-manifold} for a choice of $x \in 
    \mathbb{R}^{|E|-|V| + 1}$.
\end{prop}

\paragraph{Basis for the shear subspace.}
We now choose an explicit basis $\{\lambda^{\perp, ij}\}$ for $\mathcal{S}$.

\begin{wrapfigure}{l}{0.25\linewidth}
    \centering
    \includegraphics{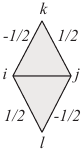}
    \label{fig:C_basis_vector}
\end{wrapfigure}

Consider two adjacent triangles $(i,j,k)$ and $(j,i,l)$, and 
assign the values for $\lambda^{\perp,ij}$ to $1/2$ for $jk$ and $li$, $-1/2$ to $ki$ and $jl$, and zero to all other edges. This set of vectors is clearly orthogonal to $\mathcal{B}$; they also have a reasonable geometric interpretation, specifically, \emph{shears (logarithms of edge cross-ratios) can be obtained as 
 $\shear_{ij} = \lambda^{\perp,ij} \cdot \lambda$}, where shears are defined as in \cite{bobenko2015discrete} by:
    \[
    \shear_{ij} = \ln \frac{\ell_{jk}\ell_{il}}{\ell_{ki}\ell_{jl}} = 
    \frac{1}{2}(\lambda_{jk} + \lambda_{il}- \lambda_{ki}-\lambda_{jl})
    \]
However, the set of vectors $\{\lambda^{\perp, ij}\}_{ij \in E}$ is linearly dependent: for any vertex $i$,
\[
\sum_{(ij) \in E } \lambda^{\perp,ij} = 0
\]
This condition is  just another form of the condition that the product of cross-ratios around a vertex is one. 
To obtain an independent set, we need to eliminate one basis vector per vertex. 
This can be done, e.g., by removing vectors corresponding to  a spanning tree of vertices along with one remaining edge that leaves the remaining $|E| - |V|$ edges connected (see appendix for a proof).  Form the matrix $C$ from the remaining columns $\lambda^{\perp,ij}$.  Note that $C \lambda$ gives the logarithms of shears of the selected edges, so the transposed pseudo-inverse $A_s = C (C^T C)^{-1} $ has the property, for any $u \in \mathbb{R}^{|V|}$, that
\[
C^T (B u + A_S s)  = s
\]
and thus the coordinates $s$ of the orthogonal space in terms of the basis $A_S$ correspond directly to shear coordinates. This gives a nice geometric interpretation for $\mathcal{S}$; however, in practice we prefer to use $C$ over $A_s$ as it is sparse and better conditioned.  Coefficients $x$ in the basis $C$, while they span the same space $\mathcal{S}$ as the independent shear basis $A_S$, are less intuitive, but can serve equally 
well as the coordinates on the manifold $F(\lambda) = 0$. We then take $S = [C \; b]$, where $b$ is the column vector corresponding to conformal scaling at the fixed vertex $v_0$ to get a complete set of $|E| - |V| + 1$ basis vectors.

\section{Distortion measures} 
\label{sec:objectives}

\subsection{Penner coordinate measures of distortion.} 
\paragraph{Log-space metric distortion.}  The simplest measure of distortion is per-edge differences in Penner coordinates 
(i.e., log-lengths in fixed connectivity); the global objective corresponding to this is
\begin{equation}
E^L = \sum_e (\lambda_e - \lambda^0_e)^2    
\end{equation}
Note that because $(\lambda_e - \lambda^0_e) = \ln \frac{\ell_e}{\ell_e^0}$,  this measure has a natural interpretation as 
one-dimensional Hencky strain energy summed over all edges. 

This distortion measure is \emph{quadratic}, which is a particularly easy case for optimization algorithms, even with nonlinear constraints. Minimizing this measure can be interpreted as finding the closest point on the constraint manifold to the original  metric. 
Below, we discuss how more common metric distortion objectives can be expressed in Penner coordinates.  

\paragraph{Area distortion.}  While metric distortion measures are most useful,
different tradeoffs between shape and area distortion may be desirable. As we explain below, 
metric interpolation, without violating angle constraints, is natural in our setting. 
In Penner coordinates, we can define a quadratic measure of area distortion, based on conformal map scale factors.

\begin{definition}
\emph{Best-fit scale factors} $u(\lambda)$ for  deformed  metric with Penner coordinates $\lambda$  relative to an initial metric $\lambda^0$,  are defined as
\[
u(\lambda) = \mathrm{argmin}_u \| (\lambda  - \lambda^0) - Bu \|^2
\]
\end{definition}
For conformally equivalent metrics, $u(\lambda)$ is the conformal scale factor relating them, which achieves the minimum value of $0$. 

The squared norm $ E^{A}(\lambda) =  \|u(\lambda)\|_2^2$ can be used as a quadratic measure of  area distortion, to which we refer as \emph{log-scale distortion.}. Likewise, the residual $\| (\lambda  - \lambda^0) - Bu \|^2$ itself is a quadratic measure of conformality, which we refer to as the \emph{log-scale residual}.

While many standard measures \emph{cannot} be naturally extended to the whole space of metrics, for fixed connectivity these can also be easily expressed in Penner coordinates, and some are even convex.  While we do not use them in our experiments, we discuss the expressions for standard ones for completeness. 

\subsection{Other distortion measures}
Pointwise measures of metric change that can be expressed as invariants of the  pointwise  metric tensor $Q$, see, e.g., \cite{Rabinovich:2017:SLI}. 

Typically, these present a trade-off between \emph{area distortion} and \emph{shape distortion}; 

These can be expressed in terms of invariants of the metric tensor $Q$, 
\[\begin{split}
J_1 &= \tr Q = \sigma_1^2 + \sigma_2^2\\
J_2 &= \det Q = \sigma_1^2\sigma_2^2
\end{split}
\]
where $\sigma_i$ are singular values of the deformation. The area distortion can be measured by how close $J_2$ is to 1 and shape distortion can be characterized, e.g., by how close $J^1/(2J_2) = \frac{1}{2}(\frac{\sigma_1}{\sigma_2} + \frac{\sigma_2}{\sigma_1})$ is to 1. In the continuum limit, conformal maps have no shape distortion.  

One natural tradeoff is given by the symmetric Dirichlet energy, that 
has the form 
\[
E^{SD} = J_1\left(1+\frac{1}{J_2}\right)
\]
it, however, goes to infinity if a triangle is degenerate (i.e., triangle 
inequality becomes equality).  The paradox is that while this is a desirable behaviour for fixed connectivity mesh optimization, this behavior is also constraining, as it may prevent the  optimization from finding a feasible solution for a set of constraints, if it requires a connectivity change.   For the same reason, it is not suitable for generalization to Penner  coordinates, as it cannot be smoothly extended to the whole space of metrics.   We show next how such generalization can be obtained. 

For a pair of  triangle meshes with the same connectivity and 
metrics $\ell^0$ and $\ell$, both invariants $J^1$ and $J^2$ can be 
expressed in terms of intrinsic per-triangle quantities
\[
J_1 = \sum_{i=1,2,3} \frac{\cot{\alpha^0_i}}{2A_0} \ell_i^2\quad
J_2 = \frac{A^2}{{A_0^2}}
\]
where $\ell_i$ are edge lengths of a triangle, $\alpha_i$ are opposite angles, 
and $A$ is the area; these quantities, in turn, have simplex expressions in terms of metric lengths only. Moreover, as shown in \cite{Chien:2016:BDP}, Symmetric Dirichlet energy, expressed in terms of (squared) lengths is convex. 

Assuming that the initial metric 
$\ell^0$ satisfies triangle inequality, both  $J_1$ and $J_2$ can be extended to arbitrary Penner coordinates for the deformed metric $\ell$:
$J_1$ is already quadratic in $\ell_i$, and $J_2$ can be expressed using Hero's 
formula as a 4th order polynomial in $\ell_i$: both are defined for 
arbitrary values of $\ell$.  However, this is not the case for energies containing 
terms like $1/J_2$, such as Symmetric Dirichlet energy. 

\paragraph{Quadratic metric distortion energy.}
As symmetric Dirichlet energy goes to infinity as triangle inequality approaches equality,  we consider the quadratic approximation to the symmetric Dirichlet 
energy in terms of logarithmic coordinate differences  $\lambda-\lambda^0$:
if $\delta = [\lambda_1-\lambda_1^0,\lambda_2-\lambda_2^0,\lambda_3-\lambda_3^0]$,
then 
\[
E^{SDQ} = \delta^T A^{SDQ}(\ell^0_1,\ell^0_2,\ell^0_3) \delta
\]
where the matrix $A^{SDQ}$ is defined as follows in terms of edge lengths (with the standard cosine law and Hero's formula
used for the cosines and the area:
\[
\begin{split}
A^{SDQ}(\ell_1,\ell_2,\ell_3) = 
\frac{\ell_1^2\ell_2^2\ell_3^2}{4\,A^4} 
&\left[ \begin {array}{ccc} 
 { {{{\ell_1}}^{2}   {\cos\alpha_1}^{2} }}
 &-{{{{\ell_1}}{{\ell_2}} \cos\alpha_3 }}
 &-{{{{\ell_3}}{{\ell_1}} \cos\alpha_2 }}
\\ 
\noalign{\medskip}
-{ {{{\ell_1}}{{\ell_2}}  \cos\alpha_3 }}
&{ {{{\ell_2}}^{2} {\cos\alpha_2}^{2} }}
&{-{{{\ell_2}}{{\ell_3}} \cos\alpha_1 }}
\\ 
\noalign{\medskip}
 -{ {{{\ell_3}}{{\ell_1}}  \cos\alpha_2 }}
&-{ {{{\ell_2}}{{\ell_3}}  \cos\alpha_1 }}
&{ {{{\ell_3}}^{2} {\cos\alpha_3}^{2} }}
\end {array} \right] 
+\\
\frac{1}{2A^2} 
&\left[ \begin {array}{ccc} 
 {{{{\ell_1}}^{4} }}
 &{ {{{\ell_1}}^{2}{{\ell_2}}^{2} }}
 &{ {{{\ell_3}}^{2}{{\ell_1}}^{2}}}
\\ 
\noalign{\medskip}
{ {{{\ell_1}}^{2}{{\ell_2}}^{2}  }}
&{{{{\ell_2}}^{4}  }}
&{ {{{\ell_2}}^{2}{{\ell_3}}^{2}  }}
\\ 
\noalign{\medskip}
 { {{{\ell_3}}^{2}{{\ell_1}}^{2}  }}
&{ {{{\ell_2}}^{2}{{\ell_3}}^{2} }}
&{ {{{\ell_3}}^{4} }}
\end {array} \right] 
\end{split}
\]
By construction, this energy has the property that it coincides with symmetric Dirichlet for small metric deformations. 
As it is expressed in intrinsic variables and has the permuational symmetry in edge indices, it is geometrically invariant and isotropic.
Finally, it is quadratic, hence convex,  and defined on the whole space of metrics.
This is the main energy we use for our examples. 

\section{Computing gradients}
\label{sec:gradients}
The last element we need is computation of the gradients of the distortion measures  described in the previous section, to be used in gradient-based optimization algorithms. 

 Using the parametrization of the constraint manifold, we write the objective as $E_S(x) = E(\lambda_S(x))$, where $E$ is any Penner coordinate distortion energy from Section~\ref{sec:objectives}. By the chain rule, the derivative of the energy is reduced to the derivative of $u_S(x)$:
\[
\nabla_{x} E_S = \nabla_{\lambda} E \cdot \nabla_{x} \lambda_S = \nabla_{\lambda} E \left( S + B\nabla_{x} u_S \right)
\]
Using the defining implicit equation $F(S x + u_S(x)) = 0$, we reduce computing this derivative to an equation
\[
0 = \nabla_{x} \left[ F(S x + u_S(x)) \right] = \left( \nabla_{\lambda} F \right) \cdot \left( S + B\cdot \nabla_{x} u_S \right)
\]
 $\nabla_{\lambda} F B$ is known to be invertible (if one scale factor $u$ is eliminated), as this is known to be the standard cotangent matrix.  From this we obtain
\[
\nabla_{x} u_S = -(\nabla_\lambda F \cdot B)^{-1} \nabla_\lambda F \cdot S
\]
Putting the two equations together, we have that
\begin{mdframed}
\[
\nabla_{x} E_S = \nabla_{\lambda} E \left( S - B (\nabla_\lambda F \cdot B)^{-1} \nabla_{\lambda} F \cdot S\right)
\]    
\end{mdframed}
reducing the energy derivative computation to derivatives of $E$ with respect to $\lambda$, which are straightforward to obtain in all cases, and $\nabla_\lambda F$, the derivative of the angle constraints. 

\paragraph{Computing $\nabla_\lambda F$.}
From \eqref{eq:constrained_manifold},
\[
\DF(\lambda) = S\nabla_{\tl} \alpha\, \nabla_\lambda\,Del(\lambda)
\]
where $\tl$ are the log length coordinates for the Delaunay connectivity.
The calculation of derivatives of $\alpha$ is standard; for completeness, we include the formulas in the appendix.  The main part of the computation is the Delaunay derivative  $\nabla_\lambda \del(\lambda)$. 
Define $\tl^i = \tau^i(\tl^{i-1})$, $\tl^0 = \lambda$.
 
Conceptually, the matrix $\nabla_\ell \del$ can be computed as a  product of differential matrices $D^i = \nabla_{\lambda} \tau^i$ for each flip. Interestingly, the entries  of this matrix can be expressed entirely in terms of \emph{shear} corresponding to the flipped halfedge. If $\tau$ is the map corresponding to flipping edge $e$, with for incident edges $a,b,c,d$, defined as in Figure~\ref{fig:intrinsic}, define  the standard shear
\[
t = \frac{\ell_a \ell_c}{\ell_b\ell_d} = e^{\shear_e}
\]
where $\shear_e$ is the shear coordinate introduced in Section~\ref{sec:projection}.
Shear is a discrete conformal invariant.
The matrix of  derivatives of $\tau(e)$ with respect to $\lambda$, 
is an identity matrix, except the rows corresponding to $e$, 
which is zero except the subrow corresponding to 
edges $e,a,b,c,d$, which has the form:
\[
D_{e,[e,a,b,c,d]} = 
\left[-2,
\frac{2t}{1+t}, \frac{2}{1+t}, 
\frac{2t}{1+t}, \frac{2}{1+t}, 
\right]
\]
As each matrix $D^i$ has a single row different from the identity matrix  and only five entries are non-zero, the computation of the matrix product 
$D A$ needed for each step of building $\nabla_\lambda \del$ 
is very low cost, so for even for a large number of flips assembly can be done efficiently.

\section{Optimization}
\label{sec:projected}
\label{sec:unconstrained}
\paragraph{Optimization in shear coordinates}
The most direct approach is to apply a gradient-based optimization method 
(in the simplest case gradient descent, but more efficiently, conjugate gradient or L-BFGS).  In shear coordinates, the problem is conceptually an unconstrained optimization problem  of minimizing $E(\lambda_S(x))$ with respect to $x$, 
with $E$ differentiable with respect to $x$, as demonstrated in Section~\ref{sec:gradients} by explicitly computing its gradient. 

As a consequence, standard convergence analysis for gradient-based methods applies. 
E.g., for any starting point, the method is guaranteed to be arbitrary close to an extremal point where $\nabla_x E = 0$. It is also likely, although we do not verify this formally, that for almost all starting points, this point is a local minimum \cite{nemirovskioptimization1999}.  While convergence to saddle points is possible we have not observed it in practice. 

An example of the energy landscape with respect to shear coordinates is shown in \ref{fig:landscape}.  This is the landscape for a tetrahedron for which $S$ has dimension 2.   We found that energy smoothness strongly depends on the assignment of target angles and connectivity, with considerable impact on the range of convergence. 

\begin{figure}
    \includegraphics[width=\columnwidth]{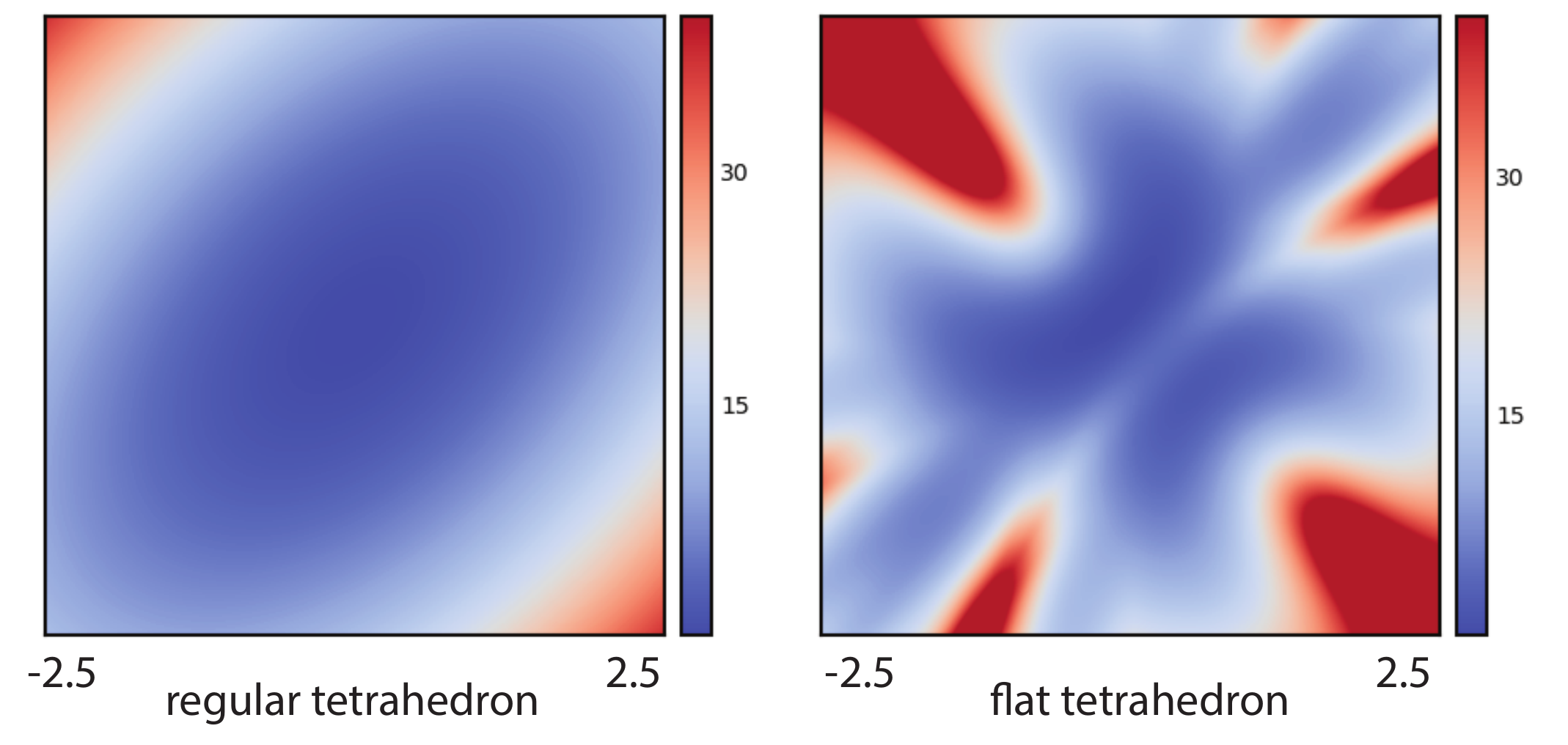}
    \caption{Log length energy landscapes for shear optimization of a tetrahedron with different angle constraints: regular tetrahedron has all vertex angles equal, and "flat" tetrahedron has $2\pi$ total angle at one vertex and equal $2\pi/3$ at the remaining ones.  Note that the second landscape is considerably less smooth.
    }
    \label{fig:landscape}
\end{figure}

\paragraph{Coordinate-projected gradient-based optimization}
While the unconstrained optimization in Section~\ref{sec:unconstrained} is appealing in its conceptual simplicity and theoretical guarantees, we found that the convergence of most methods is slow, and can be substantially accelerated by working directly in Penner coordinates and using coordinate projection to return the variables to the constraint manifold at each step. 
In this case we formulate the problem as 
\[
\min_\lambda E(\lambda),\;\mbox{subject to}\; F(\lambda)=0 
\]

The  step $\Dl$ in Penner coordinates at iteration $k$, before projection to the constraint, is computed as follows 
\[
\Dl^k = -\beta^k \left(\nabla E(\lambda^k) + \DF_k^T \mu^k\right) 
\]
where $\beta^k$ is the step size, 
$\DF_k$ is the (full rank) Jacobian of the map $F$, i.e., 
the matrix of size $(|V|-1) \times |E|$ given by 
$\DF_k = \DF(\lambda^k)$.
The vector $\mu^k$ of size $|V|-1$ is determined by the condition
\[
F_k + \DF_k \Dl^k = 0
\]
which is a linearization of the 
constraint $F(\lambda^k + \Dl^k) =0$.  More explicitly, $\mu^k$ solves the equation
\[
\DF_k \DF_k^T \mu^k = -F_k -\DF_k \nabla E (\lambda^k)
\]

The step size $\beta_k$ is determined by 
a one-dimensional line search, e.g., standard Wolfe-Armijo search; additionally, we bound the deviation from the constraint $F(\lambda) = 0$ by $\epsilon$. The update is followed by the exact projection to the constraint $\cP(\lambda)$ by solving the conformal 
map equation
\[
F(\lambda^{k+1} + Bu) = 0
\]
 for $u$ as in \cite{campen2021efficient}  while keeping $\lambda^{k+1}$ fixed,  where $u \in \bR^{|V|}$ is the log scale factors at vertices, and  $B$ is defined in Section~\ref{sec:projection}. 
 This is followed by the update $\lambda^{k+1} := \lambda^{k+1} + Bu$.
As discussed in Section~\ref{sec:projection}, this equation always has a solution. A feature of this 
solution is that $u$ do not depend on the choice of triangulation, so need not be transformed when remapped to 
Penner cell $\cP(M_0)$.

\begin{algorithm}[b]
\setstretch{0.9}
\SetAlgoLined
\DontPrintSemicolon
\SetKwInOut{Input}{Input}
\SetKwInOut{Output}{Output}
\SetKwProg{Fn}{Function}{:}{}
\SetKwRepeat{Do}{do}{while}
\SetKw{Not}{not}
\Input{
    triangle mesh $M = (V,E,F)$, closed, manifold,\newline
    edge lengths $\ell = e^{\lambda/2} > 0$ satisfying triangle inequality,\newline
    target angles $\hat\Theta > 0$ respecting Gauss-Bonnet
}
    \vspace{2pt}
\Output{
    triangle mesh $\tM = (V,\tilde{E},\tilde{F})$,\newline
    edge lengths $e^{\tl / 2}$ satisfying triangle inequality,\newline 
    with vertex angles $\max_\Theta \|\Theta - \hat\Theta\| \leq \epsilon_c$  and $\|(I- \DF  \DF F^T) \nabla E\| \leq \epsilon_m$.
    }
    \vspace{2pt}
\Fn{\scp{PennerOptimize}$(M,\lambda,\hTh)$}{
\While{ \Not \scp{Converged}$(M,\lambda)$}{
    $\tM, \tl ,D  \gets$  \scp{DiffMakeDelaunay}$(M,\lambda)$\;
    $\alpha,  \nabla_{\tl}\alpha \gets$ \scp{ComputeAnglesAndGradient}$(\tM,\tl)$ \;
     $\DF \gets \asum \, \nabla_{\tl}\alpha\, D$\;
     $L \gets  \DF  \DF^T$\;
     Solve $L \mu = \DF \nabla E(\lambda) - F(\lambda)$\;  \tcp*{Lagrange multiplier}
     $\beta \gets$ \scp{LineSearch}$(\nabla E(\lambda) + \DF^T \mu )$\; tcp*{Step size}
    $\Dl \gets -\beta \left( \nabla E(\lambda) + \DF^T \mu\right)$\; 
     $\lambda \gets \lambda + \Dl$.
     $\lambda \gets \lambda  + Bu$.
    }
    \Return $\lambda$ 
}
\vspace{1ex}
\Fn{\scp{DiffMakeDelaunay}$(M, \lambda)$}{
	$D \gets \mathrm{Id}$ \;
	$Q  \gets \{e | \mbox{\scp{NonDelaunay}}(M,\lambda, e)\}$ \; 
	\While{$Q  \neq \emptyset$}{
		remove $e$ from $Q$\;
		$\tM,\tl \gets$ \scp{PtolemyFlip}$(\tM,\tl,e)$ \;
		$D \gets$ \scp{DiffPtolemy}$(M,\tl, e) \cdot D$ 
       }
\Return $\tM,\tl,D$
}
\caption{Coordinate-projected  gradient descent algorithm summary.}
\end{algorithm}

Note that mesh connectivity $M$ does not change with iterations: at every step,  flips are done temporarily to be able to compute $\DF$, and then mesh with modified connectivity is discarded.

We can also extend the projected gradient descent to a form of projected Newton descent by solving instead for
\[
\nabla^2 E(\lambda^k) \Dl^k = -\beta^k \left(\nabla E(\lambda^k) + \DF_k^T \mu^k\right) 
\]
where $\nabla^2 E(\lambda^k)$ is the Hessian of the Penner coordinate distortion energy. With a full line step of $\beta^k = 1$, this descent direction optimizes the quadratic approximation of the energy in the tangent space to the constraint manifold. As our quadratic approximation to the symmetric Dirichlet energy has a constant sparse Hessian, $\Dl^k$ can be computed efficiently for our this energy by solving a standard quadratic programming problem. Note that this method, due to linearization of the constraints, does not have the same guarantees of quadratic convergence as unconstrained Newton descent, but it was observed to significantly improve convergence in practice for the quadratic energy.

\paragraph{Connection to shear-coordinate optimization.}
Suppose we have Penner coordinates $\lambda = S x + B u$. Any descent direction $\Delta \lambda \in \bR^{|E|}$ in Penner coordinate space can also be decomposed as $\Delta \lambda = S \Delta x + B \Delta u$. Thus, the projected gradient descent step $\lambda + \Delta \lambda + Bu'$, where $u'$ are the scale factors for the conformal projection to the constraint, decomposes as follows:
\[
\begin{aligned}
\lambda + \Delta \lambda + Bu'
& = S x + B u + S \Delta x + B \Delta u + B u' \\
& = S (x + \Delta x) + B(u + \Delta u + u')
\end{aligned}
\]
Since $u + \Delta u + u'$ must be the unique scale factors $u_S(x + \Delta x)$, we may therefore consider projected descent as just unconstrained descent in the direction $\Delta x$.

\section{Continuous Maps from Penner Coordinates}
Our algorithms produce Penner coordinates $\ell$ for the initial connectivity $M_0$ that determine a final connectivity $M$ with edge lengths $\ell'$ via $(M,\ell') = \text{Del}(M_0,\ell)$ defining a Euclidean metric on $M$.
One of the strengths of our method is it postpones the need to merge the final and updated connectivity explicitly, working in a simple optimization space with variables associated with the edges of the fixed triangulation $M_0$.

However, many applications require maps between the initial
\begin{wrapfigure}{l}{0.25\linewidth}
\begin{overpic}[width=1.3\linewidth]{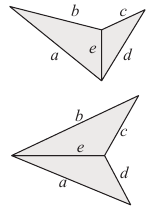}
\end{overpic}
\end{wrapfigure}
surface $PL(M_0, \ell^0)$ and the final surface $PL(M,\ell')$ associated with these cone metrics.
If the cone metric $(M_0, \ell)$ also lies in the Penner coordinate cell $\cP(M_0)$, then piecewise affine maps for the triangles of $M$ can be inferred from the edge lengths determined by $\ell^0$ and $\ell$.  In general, however, the triangulations $M_0$  and $M$ will differ; moreover, it will also generally not be possible to intrinsically flip $(M_0, \ell^0)$ to the final connectivity or vice versa while preserving Euclidean lengths.
The inset shows a simple example where the connectivities differ and the edge $e$ cannot be flipped in either mesh.
 
For the specific case of conformally equivalent metrics, \cite{campen2021efficient} describes an scheme for defining a continuous map from the initial mesh to the final that is well-defined even if $M_0 \neq M$, using circumcircle-preserving projective maps. 
Unfortunately, this method does not extend directly to the general setting; even for a fixed triangulation, circumcircle-preserving projective maps will not be continuous across edges of adjacent triangles if two metrics are not conformally equivalent  \cite{springborn2008conformal}.

\paragraph{Summary of the approach.}  Our overall approach is to reduce the problem for PL meshes to the problem of mappings between triangulated surfaces equipped with \emph{decorated ideal hyperbolic} metric that correspond to the Penner coordinates $(M,\ell)$ of the metrics we consider,  as we describe in more detail below. 
The distinctive feature of these hyperbolic metrics 
is that they are well-defined for any choice of Penner coordinates, whether this choice defines a Euclidean metric or not. We denote the associated p.w. hyperbolic surface by $H(M, \ell)$.

It is known how to map between the Euclidean surface $PL(M, \ell')$ and the decorated ideal hyperbolic surface $H(M, \ell')$ when $\ell'$ is Delaunay with respect to $M$ \cite{campen2021efficient, gillespie2021discrete}, and we decompose the mapping problem between $H(M_0, \ell^0)$ and $H(M, \ell')$ into two steps. 
\begin{enumerate}
\item We keep the \emph{hyperbolic metric} on $H(M,\ell')$ fixed,  and change the connectivity to $M_0$, obtaining an ideal hyperbolic surface $H(M_0,\ell)$  by retriangulating $H(M,\ell')$; 
\item We keep the  \emph{connectivity} fixed and map $H(M_0,\ell^0)$ to $H(M_0,\ell)$.
\end{enumerate}

For the first step, we observe that \cite{campen2021efficient}, relies only on the hyperbolic metric on $M_0$ and $M$, which we discuss below,  so it can be applied to construct the map  $H(M_0,\ell)$ to  $H(M,\ell')$  without any changes, despite the fact that $(M_0, \ell)$ does not necessarily define a Euclidean metric.

Briefly, this algorithm constructs an overlay mesh $M^r$ which is a refinement of both $M_0$ and $M$, with edges of $M_0$ and $M$ embedded in $M^r$ as sequences of edges, with positions of newly inserted vertices defined by the shared hyperbolic metric on $M_0$ and $M$ (hyperbolic metrics implied by $\ell$ and $\ell'$ are identical in the case of discretely conformally equivalent $(M_0,\ell)$ and 
$(M,\ell')$).  The final map is defined by a continuous p.w.projective  map on triangles of $M^r$. We refer to \cite{campen2021efficient,gillespie2021discrete} for the details of the algorithm. 

For the second step, we show how one can define a continuous p.w. projective map between two \emph{different} ideal hyperbolic metrics defined on the same triangle mesh $M$, with Penner coordinates $\ell^0$ and $\ell$.

\subsection{Ideal hyperbolic metrics on triangulations}
Suppose we are given a mesh $M$ with Penner coordinates $\ell$ assigned to its edges.

To define the metric on each triangle of a mesh $M$, we make use of the Beltrami-Klein hyperbolic plane model, which represents the hyperbolic plane using a unit disk. In this model, chords of the disk are straight lines, and the isometries of the model are precisely the projective maps of the Euclidean plane that preserve the unit disk.  We identify each triangle of the mesh $T$ to a triangle inscribed in a unit disk of the Beltrami-Klein model. In the hyperbolic metric, 
this triangle is an \emph{ideal triangle}, as its vertices are on the boundary of the circle, i..e,  correspond to points at infinity in hyperbolic metric (ideal points).  The triangle sides are infinite straight lines.  All ideal triangles are congruent:   there is a (unique) circumcircle-preserving projective map that carries a given triangle onto any other, and these maps are isometries in the hyperbolic metric. 
Note that so far, we did not make use of Penner coordinates $\ell$.

To construct a complete metric on $M$, we need to identify the sides of the ideal triangles corresponding to the same edge.  Because the lengths of the sides are infinite, there is a one-parametric family of isometries between these sides: the  isometries differ by a translation (\emph{shear}). 

Formally, for an edge $e_{ij}$ in a pair of adjacent triangles with adjacent edges with logarithmic Penner coordinates $\lambda_a, \lambda_b, \lambda_c, \lambda_d$, the corresponding shear coordinate $\shear_e$ is given by 
\begin{equation}
\shear_e = \frac{1}{2}\left(\lambda_a - \lambda_b + \lambda_c - \lambda_d\right)
\label{eq:shear}
\end{equation}
Note that these are exactly the same shears that are used in Section~\ref{sec:projection}.

Two choices of logarithmic Penner coordinates $\lambda, \lambda'$ describe the same hyperbolic metric if and only if the corresponding shear coordinates $\sigma, \sigma'$ are the same. Conversely, for any $\sigma \in \bR^{|E|}$ such that the sum of values around any given vertex is zero, there is some hyperbolic metric with shear coordinates $\sigma$.   Moreover, if shears are the same, and  Penner coordinates $\ell_{ij}  = e^{\lambda_{ij} / 2}$ satisfy triangle inequalities, then the corresponding \emph{Euclidean} metrics are conformally equivalent. 

The  formula \eqref{eq:shear} is not arbitrary: $\lambda_{ij}$ have a natural interpretation in terms of lengths of segments cut out from the infinite edges by \emph{horocycles} (curves specific to hyperbolic geometry, which are limits of circles with a common tangent at a point as the radius goes to infinity). 
This connection is explained in detail in \cite{springborn2019ideal} and \cite{gillespie2021discrete}.  

To define the the metric on the whole surface precisely, we define charts on pairs of adjacent triangles 
$T_{ijk}$ and $T_{jil}$ sharing edge $e_{ij}$.  Following \cite{campen2021efficient}, each chart is a pair of congruent right-angle triangles  $(T_1, T_2)$ with one side of each, corresponding to the edge $e_{ij}$ spanning the horizontal diameter of a unit circle, forming a rectangle (Figure~\ref{fig:chart}). 
The positions of the vertices $p^r = (r, \sqrt{1-r^2})$ and $-p^r$ at the right angles  of the triangles $T_1$ and $T_2$ are chosen so that the hyperbolic distance in the Beltrami-Klein model between their projections on the horizontal diameter is the shear $\shear$. Explicitly, they are given by $r = (1 - c)/(1 + c)$, where $c = e^{\shear}$, and the shear coordinate for the common edge in the two triangle chart is also $\shear$.

\begin{figure}
    \centering
    \includegraphics{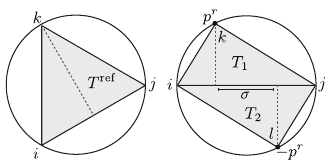}
    \caption{Two-triangle chart for the ideal hyperbolic metric construction.}
    \label{fig:chart}
\end{figure}

If two distinct charts overlap, the overlap is always a single triangle; suppose this triangle corresponds to $T_1$ in the first chart and $T_1'$ in the second chart.  We define the transition maps between two overlapping charts to be the unique circumcenter-preserving map between $T_1$ and $T_1'$.  The only cycle condition that needs to hold for two-triangle charts is that the composition is identity for a cycle of charts around a vertex, which is ensured by the fact that shears add up to zero around a vertex.  As the transition maps are isometries on the overlaps this set of charts defines a hyperbolic metric on $M$.

\subsection{Maps between hyperbolic surfaces}

\paragraph{Projective maps between surfaces}
Let $H$ be the set of oriented halfedges of $M$, i.e., sides of triangles of $M$, with two sides associated with each edge for closed meshes.  Let $\tau \in \bR^{|H|}$ be such that for any triangle $T_{ijk}$ these values satisfy 
\[
\tau_{ij} + \tau_{jk} + \tau_{ki} = 0
\]
and for any edge $\{ij, ji\}$ they satisfy
\[
\tau_{ij} + \tau_{ji} = \shear_{ij}' - \shear_{ij}
\]
These constraints can be expressed by 
\[
\begin{aligned}
Z\tau & = 0 \\
C\tau & = \shear' - \shear
\end{aligned}
\]
where $Z$ is a $|F| \times |H|$ matrix, $C$ is a $|E| \times |H|$ matrix, and $\shear$ and $\shear'$ are vectors of per edge shears for $H(M, \ell)$ and $H(M, \ell')$ respectively. We then have the following proposition:

\begin{prop}
Suppose two hyperbolic surfaces $H(M, \ell)$ and $H(M, \ell')$ are defined with respect to the same connectivity, and suppose $\tau \in \bR^{|H|}$ satisfies the linear constraints
\[
\begin{aligned}
Z\tau & = 0 \\
C\tau & = \shear' - \shear
\end{aligned}
\]
We equip each triangle $T_{ijk}$ in $M$ with an equilateral reference triangle $T^{\text{ref}}$, and we define the map $P^{\tau}_{ijk}: T^{\text{ref}} \to T^{\text{ref}}$ in terms of (unnormalized) barycentric coordinates on the equilateral reference triangle by
\[
P^{\tau}_{ijk}(w_i, w_j, w_k) = 
\begin{bmatrix}
e^{2(\tau_{ki} - \tau_{ij})/3}w_i \\
e^{2(\tau_{ij} - \tau_{jk})/3}w_j \\
e^{2(\tau_{jk} - \tau_{ki})/3}w_k
\end{bmatrix}
\]
We then define the map $P^{\tau}: H(M, \lambda) \to H(M, \lambda')$ per triangle by
\[
P^{\tau}|_{T_{ijk}} = P^{\tau}_{ijk}
\]
$P^{\tau}$ is well-defined and continuous with respect to ideal hyperbolic surfaces  $H(M, \ell)$ and $H(M, \ell')$
\end{prop}

A proof of this proposition is provided in the appendix.

\paragraph{Minimal translation map}
Note that, unlike in the affine and conformal case, the continuous projective map is not generally unique. There are $3|F|$ degrees of freedom and only $|E| + |F|$ constraints. We choose to use the projective maps that minimize $\frac{1}{2}\|\tau\|_2^2$ subject to the constraints. This has a geometric interpretation as minimizing the sum of squares of hyperbolic translations along the edges of the triangulation (see appendix). In the case where $\shear = \shear'$, i.e., the initial and final metrics are conformally equivalent, the unique solution $\tau = 0$ to the minimization will then describe an isometry and thus reduce to \cite{campen2021efficient}. Since $\frac{1}{2}\|\tau\|_2^2$ is a quadratic energy and our constraints are linear, the solution is easily found. We also note that this construction closely resembles the halfedge forms defined in \cite{custers2020}.

\paragraph{Refinement} Our method produces a map $f:PL(M_0, \ell^0) \to PL(M,\ell')$ that is continuous and p.w. projective on the common overlay mesh $M^r$ generated in the first step. This overlay can be used directly to generate refinements $PL(M_0, \ell^0)^r$ and $PL(M,\ell')^r$ of the PL meshes so that their triangular faces are in one-to-one correspondence and $f$ is a projective bijection between corresponding faces. However, the resulting refinements can be unnecessarily fine, often producing an increase of up to $200\%$ of the original face count. We adapt the method of \cite{weber2014locally} to instead generate a coarse common triangulation of the two cone metrics that only performs local refinement where it is necessary. In brief, we use $f$ and a planar embedding of $PL(M, \ell')$ to map each face $\triangle$ of $(M_0, \ell^0)$ to a triangle $f(\triangle)$ in $(M, \ell') \subset \mathbb{R}^2$. If all triangles are mapped to the plane with positive orientation, then this procedure generates a compatible retriangulation of $PL(M, \ell')$ with no refinement, and we use the natural p.w. linear map for this common triangulation. Otherwise, the inverted faces are refined in $PL(M_0, \ell^0)$, inserting vertices in the same positions as in the overlay mesh, and the new subfaces are again mapped to the plane by $f$. We proceed iteratively until no faces are inverted. This refinement will in the worst case recover the original overlay mesh and is thus guaranteed to terminate with a valid triangulation.

\section{Experiments}\label{sec:results}

In this section we demonstrate the performance of the method on the dataset from \cite{Myles:2014}, 
in several forms, similar to \cite{campen2021efficient} metric optimization on closed meshes,  cut meshes with parametrization alignment to cuts and meshes with boundary. 

We compare distortion distributions using  best-fit conformal factors, to measure area distortion, and symmetric
ratios of 3D to parametric edge length (edge stretch factors) to measure metric distortion, in a way that is independent and distinct from  the metric distortion objectives we use. 
  
Our intrinsic approach, as discussed in Section~\ref{sec:related} is complimentary in terms of problem formulations to the approaches formulated in terms of $(u,v)$ variables.  The only type of existing methods that is closely aligned with our setting are conformal maps. 

For example, a symmetric Dirichlet  parametrization of a topological disk naturally supports free boundary conditions or fixed boundary.  Natural conditions for our method are prescribed boundary angles, or isometric mapping on boundary edges as in \cite{springborn2008conformal}.  
Nevertheless, to provide a quantitative comparison between optimizing linearized symmetric Dirichlet energy $E^{SDQ}$ in Penner coordinates and standard symmetric Dirichlet energy, we compare the distribution of 
stretch factors in similar, but not identical scenarios.

Unless otherwise noted, the textures and shadings of the surface are the same as in Figure \ref{fig:teaser}.

\paragraph{Metric interpolation.}
One useful feature of Penner coordinates is that one can easily interpolate
between any two metrics by interpolating their coordinates, as there is no need to avoid triangle inequality violations.  For cone metrics with prescribed angles, the interpolated metric need not have these angles, but an additional conformal projection can be performed on the interpolated coordinates. 
For example, we can achieve any desired tradeoff by computing a metric with Penner coordinates $\lambda^A$, by optimizing the scale distortion measure $E^A$, and pure conformal map,  which minimizes shape distortion, with Penner coordinates $\lambda^C$, 
and a metric $\lambda^I$ minimizing isometry measure $E^{SDL}$;
then we can compute any intermediate metric with prescribed angles $\hat\Theta$ as 
\[
\lambda^w = \cP(w_1\lambda^A + w_2\lambda^C +w_3\lambda^I)
\]
where $w_1 + w_2 + w_3 = 1$, or use another type of interpolation between these two metrics.  Figures  \ref{fig:teaser} and  \ref{fig:interpolation} show examples of such interpolation.

\paragraph{Optimization objectives.}  We consider three optimization objectives described Section~\ref{sec:objectives}.  Figure~\ref{fig:objectives} shows the results of optimizing a number of these on two models. 
For each model, we include histograms showing the distribution of different distortion measures: 

The difference between conformal and other measures
is obvious in much higher scale distortion. 
The difference between $L_2$ and $L_p$ ($p=4$) 
version of the log-length energy is more subtle: 
one can observe that $L_p$ penalizes more local 
distortion but slightly increases the average, as expected.  Finally, log-scale distortion has
the lowest distortion of scale at the expense of 
more anisotropic stretching. A more severely distorted
parametrization is shown in Figure~\ref{fig:interpolation}, with interpolation 
between conformal, log-length and log-scale,  showing some of these effects more clearly.
E.g., note how for the conformal map scale residuals are all at zero, but there is a broad distribution of scale factors. 
At the opposite extreme, for log-scale measure, the scale factors are all at zero, i.e., there is virtually no area distortion, but there is broad range of scale residuals. 

\begin{figure}
    \includegraphics[width=\columnwidth]{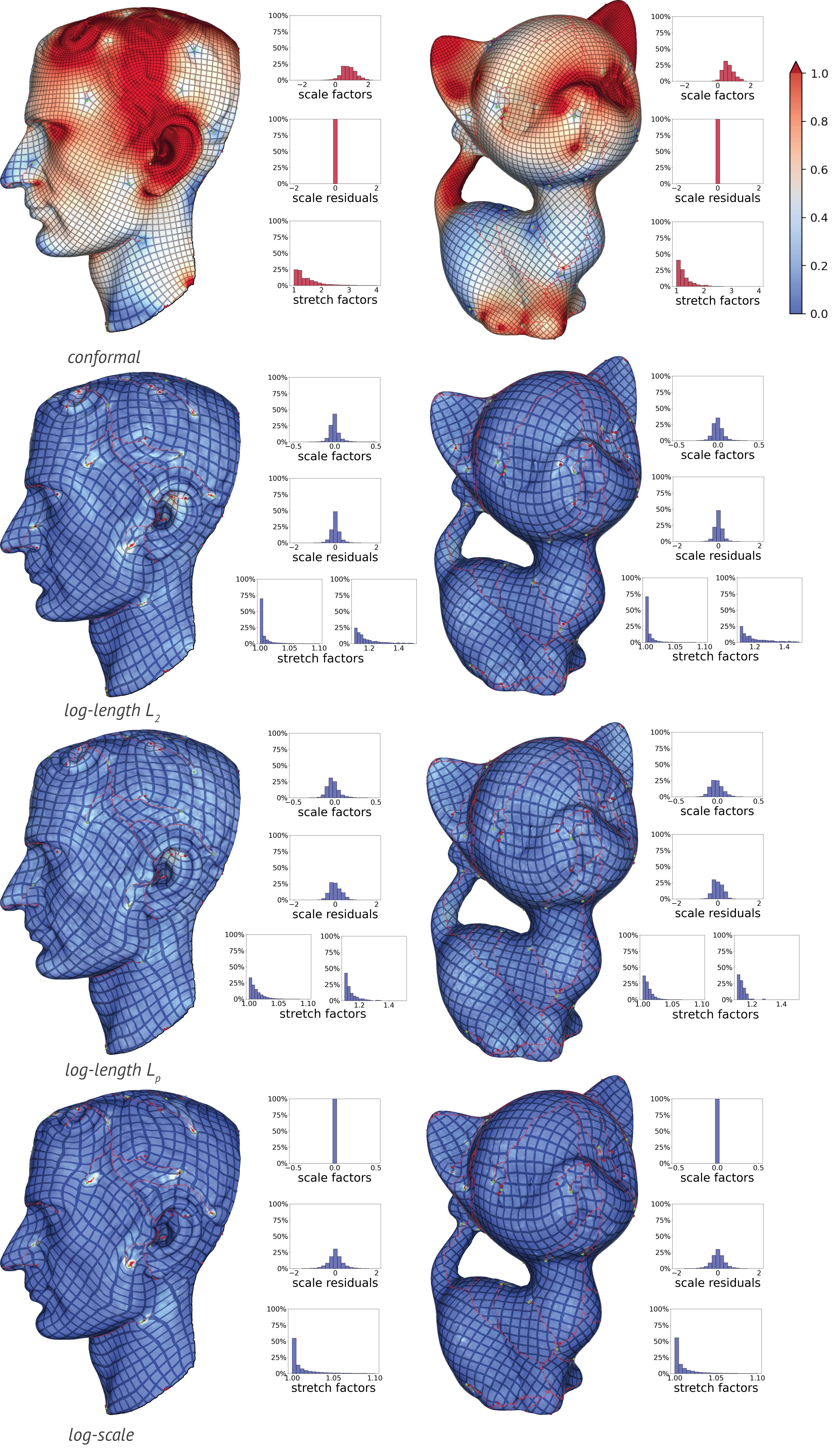}
    \caption{Optimization results for different objectives on an open mesh (left) and a closed mesh (right). We also visualize the distribution of best fit scale factors, scale residuals, and stretch factors (Section~\ref{sec:objectives}). All distortion objectives lead to approximately isometric parametrizations that are visually similar, but the distributions of the measures differ significantly.
    }
    \label{fig:objectives}
\end{figure}

\begin{figure*}[t!]
\centering
\includegraphics[width=\textwidth]{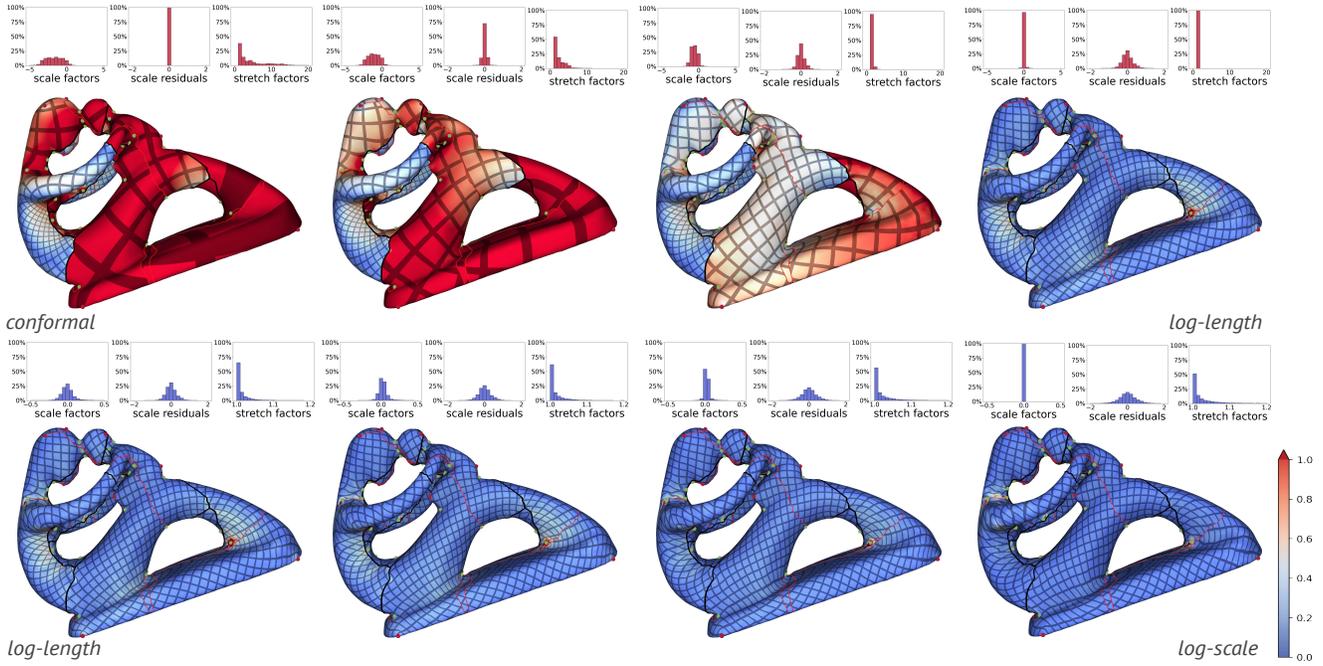}
    \caption{
    Top row: Interpolation from conformal to minimizer
    of log-length distortion $E^I$. Bottom row: interpolation between minimizers of log-length distortion $E^I$ and log-scale distortion $E^S$. This is a cut mesh with the cut boundary marked by black paths. As the metric is interpolated, the distributions of scale factors, scale residuals, and stretch factors also vary consistently, with the scale distortion decreasing, and length distortion increasing. 
    }
    \label{fig:interpolation}
\end{figure*}

\paragraph{Sensitivity to initialization.} In general, our method has a natural starting point (the original metric).  However, as the flatness constraints that we use are not convex, in general, the local constrained minimum we found is not unique. 
We vary the starting point of the optimization process for the $E^{SDQ}$ quadratic distortion measure, 
(Figure~\ref{fig:initial}), by adding  perturbations to the initial lengths. Note that as Penner coordinates have no constraints adding arbitrary noise still yields a valid starting point.  We observe that the results are quite similar even for very high noise. 

\begin{figure}
\includegraphics[width=\columnwidth]{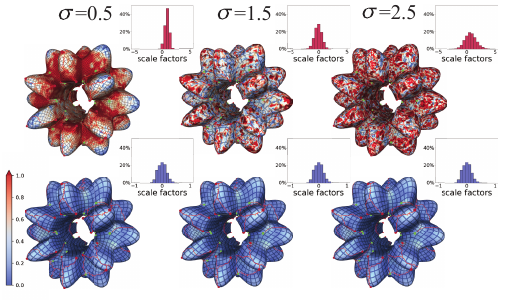}
    \caption{Varying the initial point for optimization: the numbers indicate additive Gaussian noise perturbation of logarithmic edge lengths; i.e, perturbation 2.5 means scaling by a factor of around 12. The upper row visualizes the perturbed metric using its own closest conformal projection. The lower row shows where the optimization converges in each case. }
    \label{fig:initial}
\end{figure}

\paragraph{Comparison.} We compare our approach with two other parametrization methods. The first is $uv$-optimization for a fixed connectivity with seamless constraints on a closed genus 0 mesh. We produced an initial conformal parametrization and then optimized the symmetric Dirichlet energy. The minimizer of our quadratic Penner coordinate symmetric Dirichlet energy that we compute is visually similar (Figure \ref{fig:comparison}).

Our method also supports free boundary optimization, so we compare our method with the $uv$ free boundary method for optimizing the symmetric Dirichlet energy described in \cite{Rabinovich:2017:SLI}. Some error is visible in our solution near the center of the mesh where distortion accumulates, but the overall result is again visually similar (Figure \ref{fig:disk}).

We note that some faces with large symmetric Dirichlet energy were present in both cases with our method. However, the triangle quality (as measured by the aspect ratio) remained comparable.

We also implemented a variation of our method for discrete metrics on a fixed connectivity and optimized the symmetric Dirichlet energy (as expressed in intrinsic log edge lengths) for these two examples. We successfully obtained valid solutions for both of these examples that had the same final energy (up to $10^{-7}$). However, we emphasize that this fixed connectivity approach is not guaranteed to find a solution (and indeed in many cases did not produce a valid output).

\begin{figure*}[t!]
    \includegraphics[width=0.66 \textwidth]{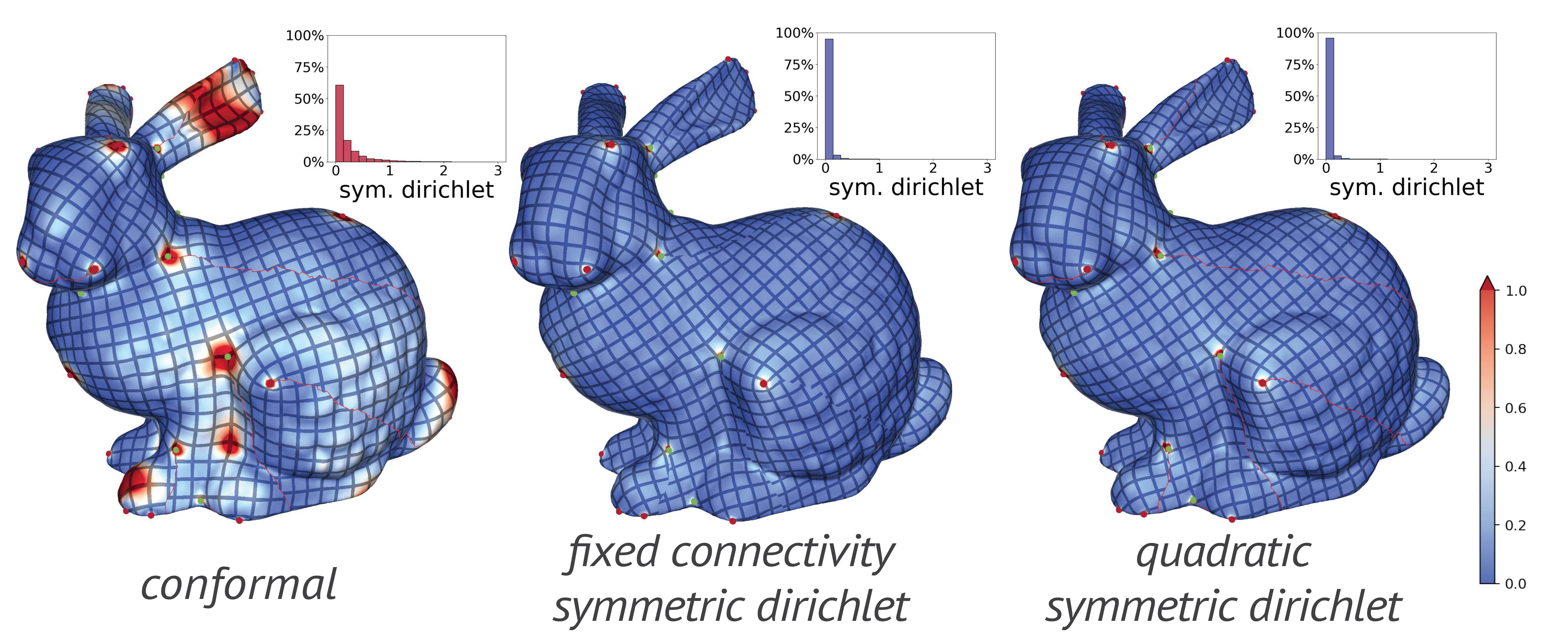}
    \caption{Optimization on a closed mesh compared with $uv$ optimization of the symmetric Dirichlet with seamless constraints on a fixed connectivity initialized with a conformal parametrization. We visualize the symmetric Dirichlet energy as a shading on the surface.
    }
    \label{fig:comparison}
\end{figure*}

\begin{figure*}[t!]
    \includegraphics[width=0.66 \textwidth]{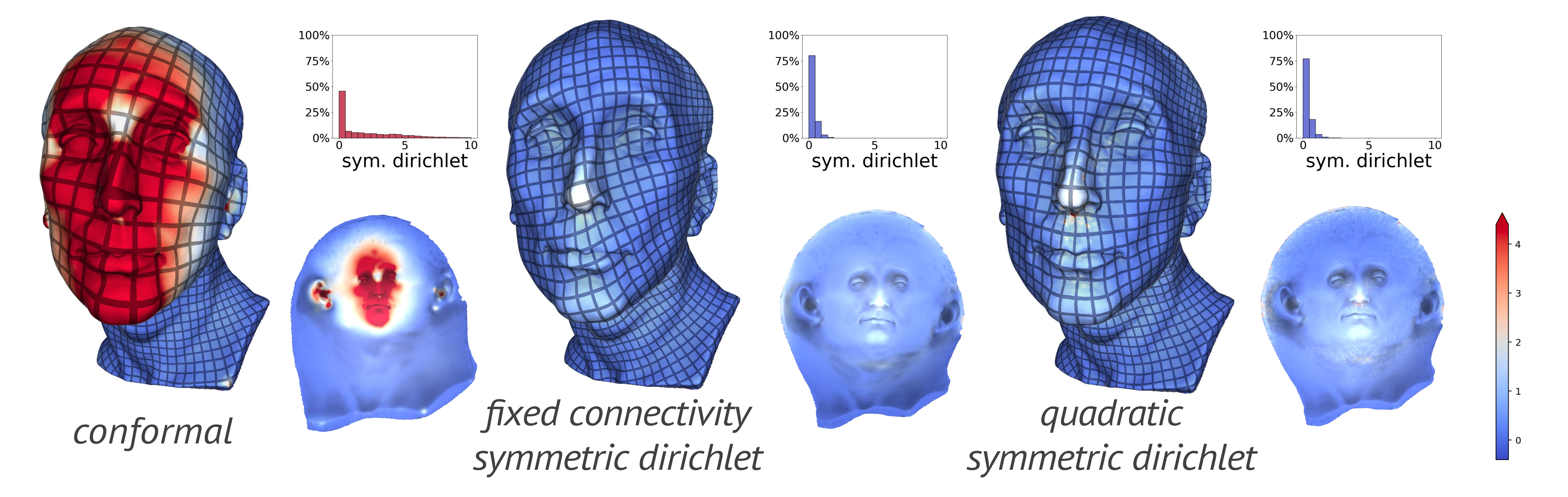}
    \caption{Optimization on a topological disk with free boundary compared with free boundary $uv$ optimization of the symmetric Dirichlet energy on a fixed connectivity initialized with Tutte's embedding. We again visualize the symmetric Dirichlet energy as a shading on the surface, and we also visualize the corresponding $uv$ domain with the shading of the 3D surface.
    }
    \label{fig:disk}
\end{figure*}

\paragraph{Testing on datasets.}  We use datasets from \cite{Myles:2014} and its cut version from \cite{campen2021efficient}, the latter being particularly challenging numerically,  leading to conformal map scale range up to $10^{100}$.  The former dataset consists of 114 meshes.  We tested with the log-length functional 
$E^I$. The distribution of average and maximal distortion measured as per-edge symmetric ratio of initial and cone-metric lengths, is shown in Figure~\ref{fig:histograms}, over all meshes in three versions of the dataset we use. 
\begin{figure}
\includegraphics[width=\columnwidth]{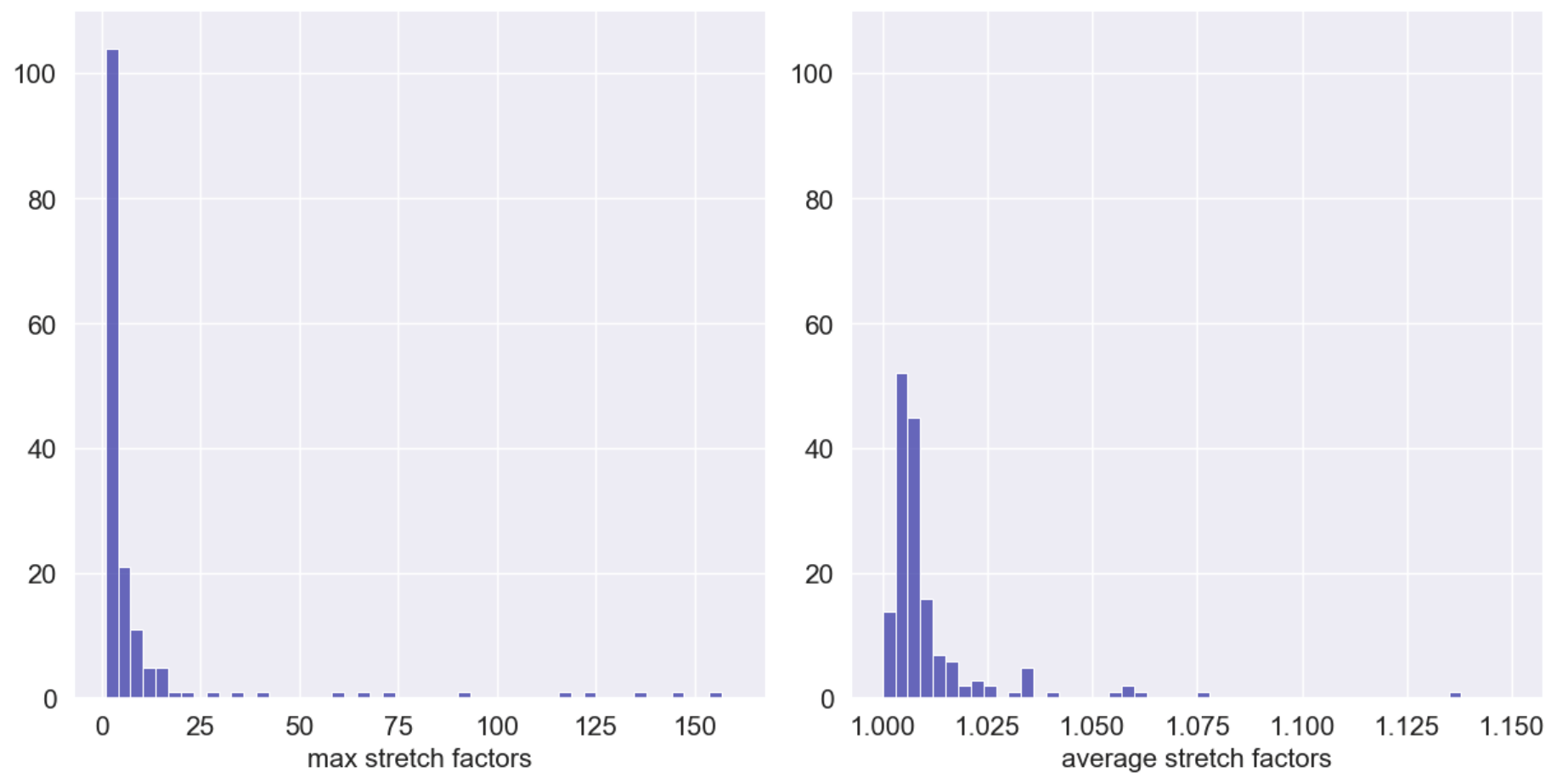}
    \caption{Left: distribution of average edge stretch factor per shape in the dataset.  Right: distribution of maximal stretch factors. Two outliers from the dataset were excluded.}
    \label{fig:histograms}
\end{figure}

\begin{figure}
\includegraphics[width=\columnwidth]{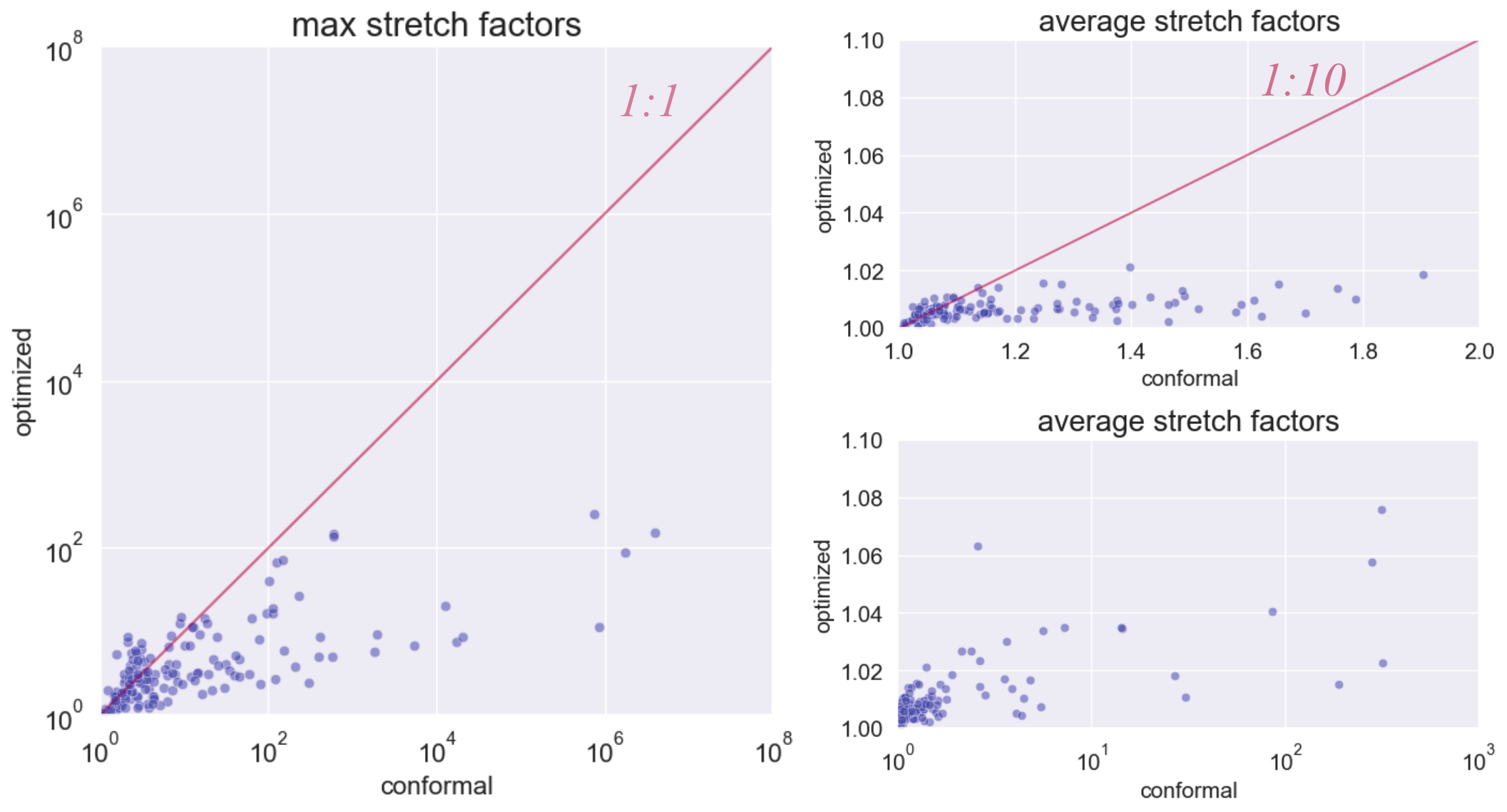}
    \caption{Left: scatter plot of maximal edge stretch factor before and after optimization.  Right: scatter plot of average stretch factors. }
    \label{fig:scatter}
\end{figure}

\begin{figure*}[t!]
    \includegraphics[width=\textwidth]{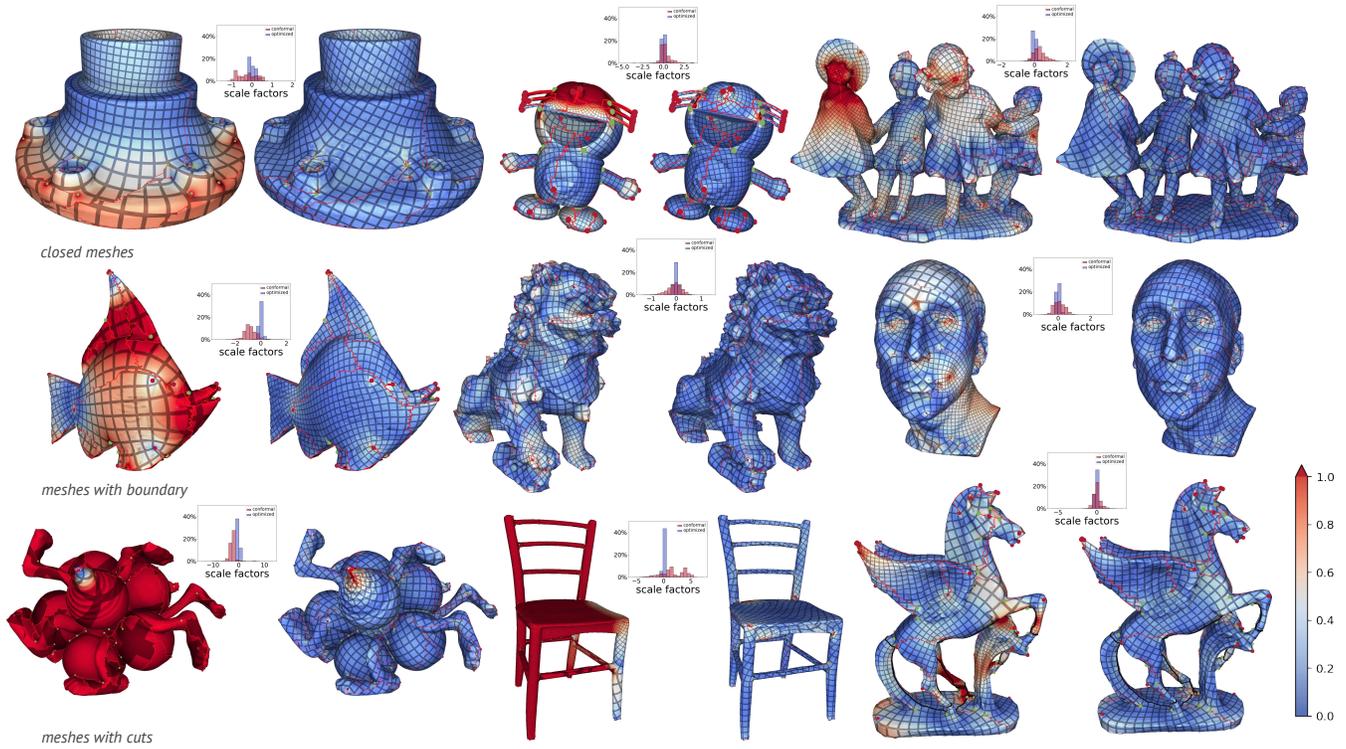}
    \caption{Examples of isometrically and conformally parameterized models: closed,  with boundary, cut with fixed boundary angles.  For  each model, we show  conformal and log-length optimized metrics. 
    }
    \label{fig:examples}
\end{figure*}

\paragraph{Methods.} We compare projected gradient and projected Newton descent with popular first-order optimization methods on the unconstrained shear space. Figure~\ref{fig:methods}
shows convergence plots comparing five optimization methods on two meshes for the log-length and quadratic symmetric Dirichlet energies. The projected gradient descent is equivalent to projected Newton descent for the log-length energy and converges quickly, but we observe that it can converge very slowly for the more poorly conditioned symmetric Dirichlet approximation. The projected Newton method, on the other hand, quickly converges for both quadratic energies, and is our primary method for optimization. For higher-order energies, such as $L_p$ log-length and log-scale distortion, the Hessian is nonconstant, so we instead use simple projected gradient descent.

\begin{figure}
\includegraphics[width=\columnwidth]{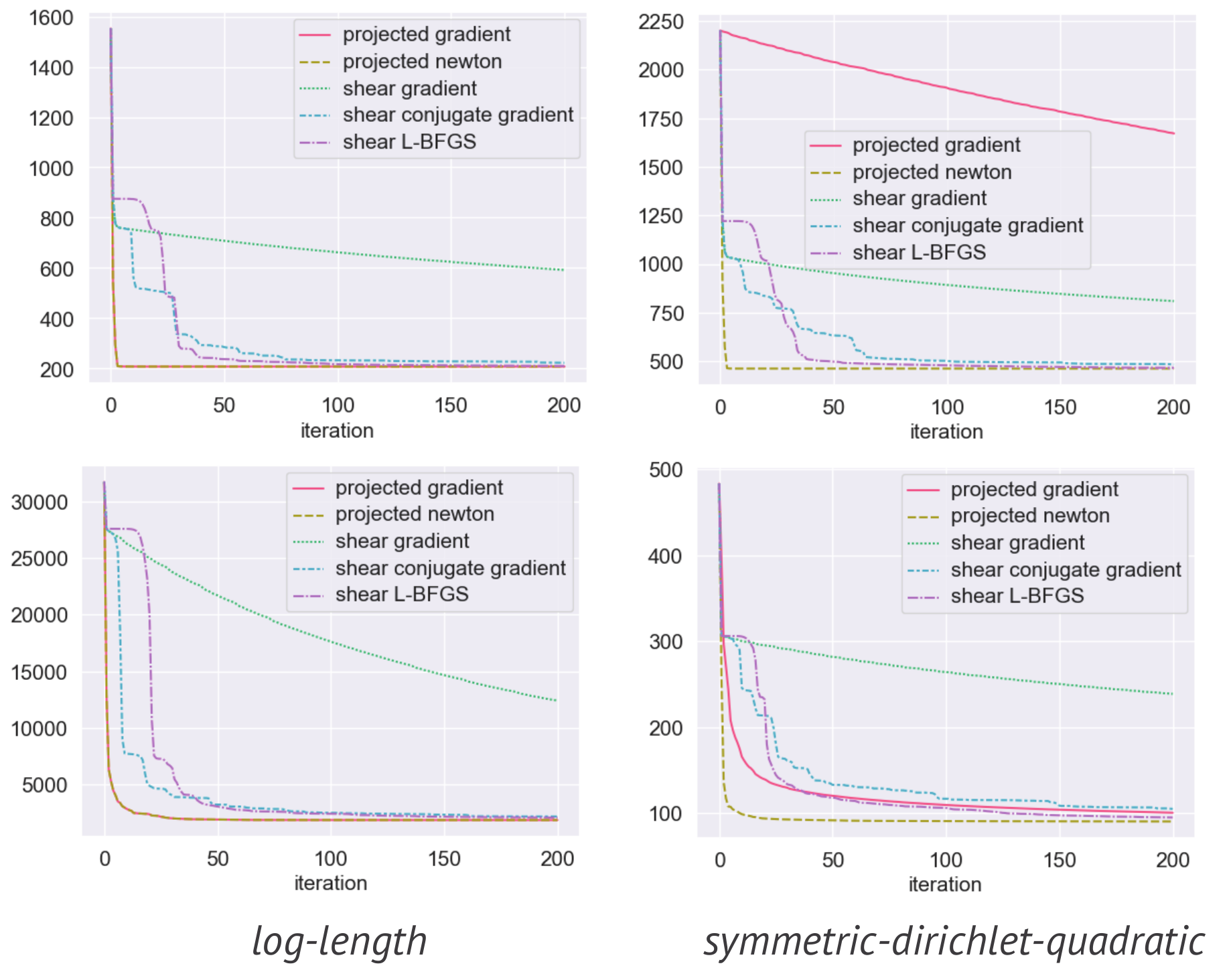}

    \caption{ Convergence plots of the energy per iteration for different optimization methods. }
    \label{fig:methods}
\end{figure}

\paragraph{Convergence.}  The algorithm overall has very stable 
behavior with most improvements on the first 10-20 iterations, except artificially challenging examples. Figure~\ref{fig:convergence}
shows convergence plots for closed and cut meshes as well as per iteration timings plotted against mesh size.

\begin{figure}
\includegraphics[width=\columnwidth]{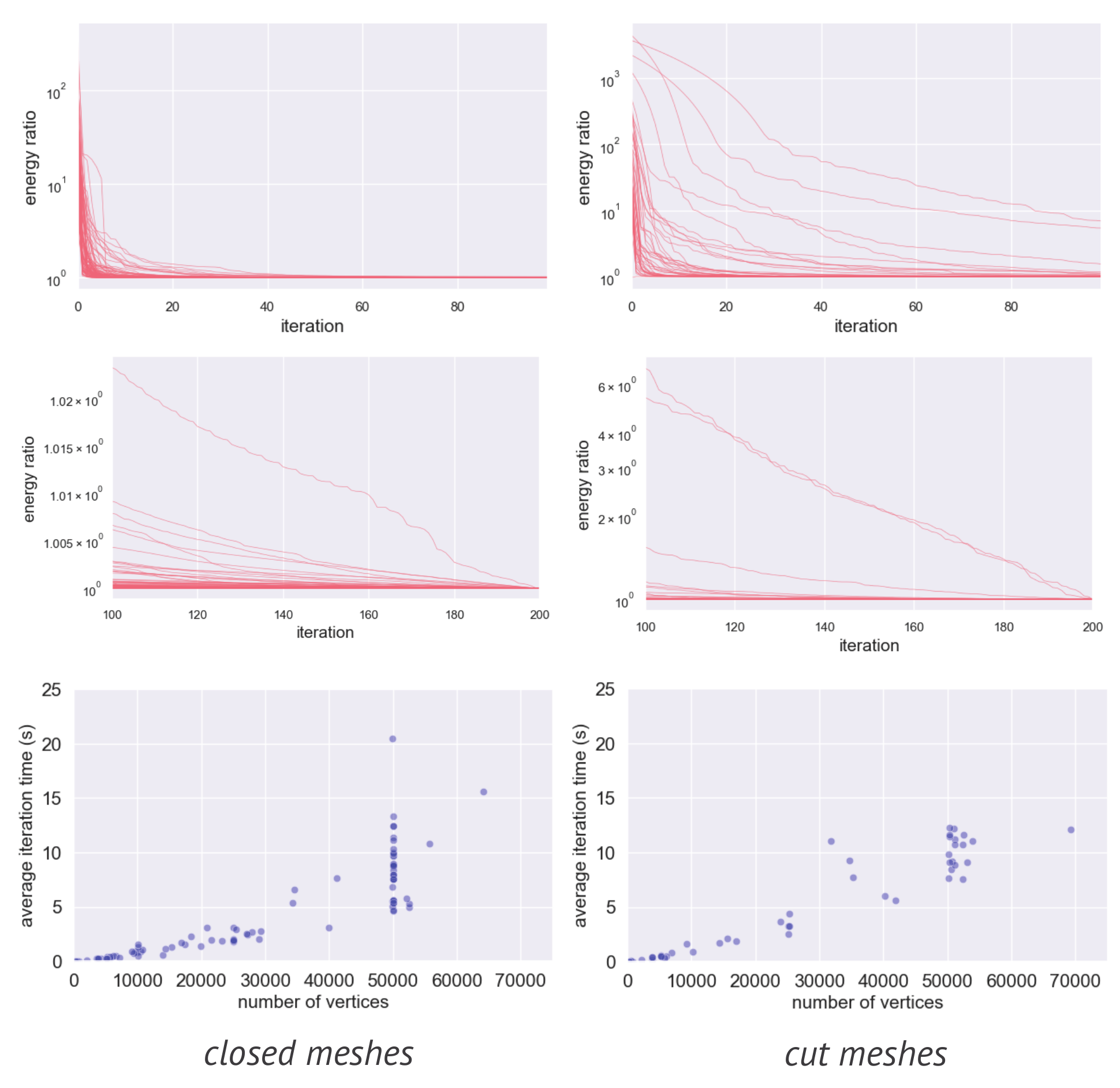}

    \caption{ Convergence plots of the ratio of the log-length energy per iteration to the final energy after 200 iterations for closed (left) and cut (right) meshes with fixed angle constraints. Orders of magnitude improvement are achieved in the first few iterations.}
    \label{fig:convergence}
\end{figure}

\paragraph{Refinement.}  Our refinement method produces much coarser refinements than the overlay, and the refined faces are localized to regions where they are necessary. Figure~\ref{fig:refinement}
shows the percentage increases in the number of faces of the refined meshes relative to the original face counts for all meshes in our data set. We also provide the percentage increases for the full overlay refinements for comparison.

\begin{figure}
\includegraphics[width=\columnwidth]{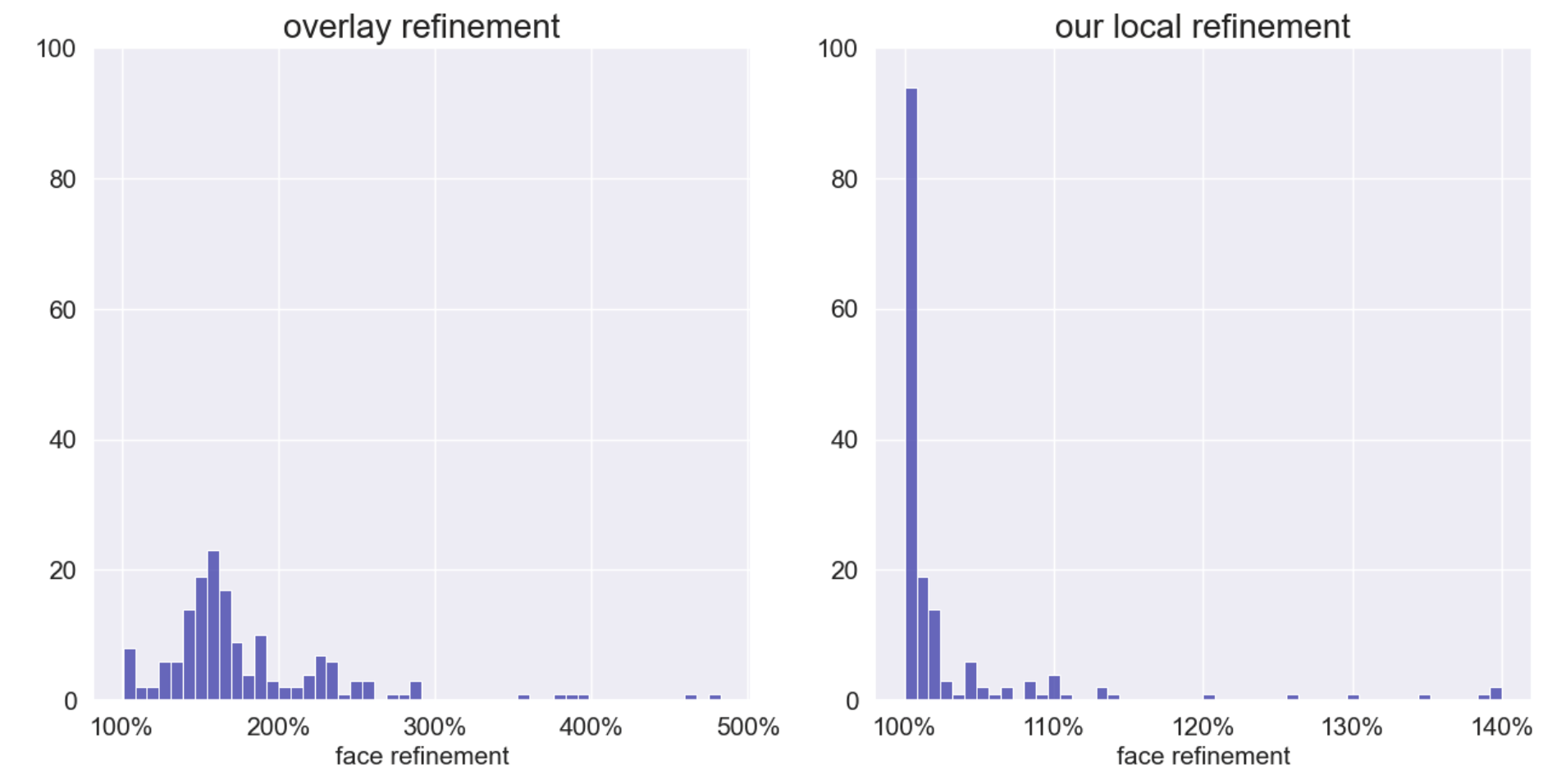}

    \caption{ Face count percentage increase histograms for the full overlay refinement (left) and our local refinement method (right). One outlier was excluded from the latter.}
    \label{fig:refinement}
\end{figure}

\section{Limitations and Concluding Remarks}
\label{sec:conculsions}

We have demonstrated that Penner coordinates yield a natural parametrization of all cone metrics with fixed vertices, allowing to express metric optimization problems in a simple form and provide guarantees on solution existence. 

We view the proposed method just as a first step in the direction of optimization of metric in Penner coordinates. As currently formulated the method is relatively slow, as first-order methods such as gradient descent and conjugate gradients converge slowly; as shown in \cite{campen2021efficient}, in most cases, the conformal map step requires a small number of iterations, but overall hundreds of linear solves may be needed.  Thus, to make the method practical, more efficient optimization methods are needed, or the the approach  needs to be combined with highly efficient techniques that do not provide guarantees, and used on already nearly flat metrics. 

While separate angle and length constraints on the boundary are natural for the proposed metric parametrization, for simultaneous angle and length constraints, equivalent, up to rigid transformations, to fixing the whole boundary, there is no proof of solution existence, and this is an interesting direction for future theoretical work. 

\paragraph{Extensions.} While we mostly focused on the most natural settings for Penner metric coordinates (closed and open meshes with fixed angles), 
the method readily extends to several settings:    fully isometric boundary with no fixed angles, as proposed in \cite{springborn2008conformal}, 
fully free boundary (no length or angles fixed), and  fully fixed boundary (both angles and lengths are fixed). 
We note that the latter case is not supported by the conformal projection, so there is no guarantee that these constraints can be 
satisfied simultaneously.  Theoretically, it is a question about intersection between a linear subspace of dimension $|E|-|E_b|$ where 
$|E_b|$ is the number of boundary edges,  and the nonlinear manifold of angle-constrained metric, of dimension $|E|-|V|$. 
It is an interesting question for future work.

\begin{acks}
Supported by \grantsponsor{_OAC}{OAC}{https://nsf.gov/}-\grantnum{_OAC}{1835712}, \grantsponsor{_CHS}{CHS}{https://nsf.gov/}-\grantnum{_CHS}{1901091}, and a gift from Adobe.
\end{acks}

\bibliographystyle{ACM-Reference-Format}
\bibliography{99-bib}

\appendix

\section{Appendix}
\label{sec:appendixA}

\subsection{Computing triangle angle derivatives}
The calculation of $\nabla_{\tl}\alpha$, where $\tl$ are
Penner coordinates for the Delaunay triangulation obtained from
$\lambda$ is standard (see, e.g., \cite{bobenko2015discrete} where it is used to compute the Hessian of the convex function in the minimization defining the conformal mapping. 

More specifically, consider the triangle with corners with indices $i,j,k$ (i.e., each angle of each triangle gets an index); and let $l,m,n$ be the indices of edges $(i,j)$, $(j,k)$ and $(k,i)$.
Then 

\begin{align*}
\partial\alpha_i/\partial \lambda_l &=& -2\cot \alpha_j\\
\partial\alpha_i/\partial \lambda_m &=&  2(\cot \alpha_j+\cot \alpha_k)\\
\partial\alpha_i/\partial \lambda_n &=& -2\cot \alpha_k\\
\end{align*}

and all other derivatives of $\alpha_i$ are zero; this determines 
the entries of $i$-th row of the $3|F| \times |E|$ matrix $\nabla_{\tl}\alpha$.

\subsection{Independent set of dual-shear basis edges}

\begin{prop} For the $|E|\times |E|$ matrix $C$ with columns given by $\lambda^{\perp,ij}$, removing columns corresponding the edges of a spanning tree $T$ of $E$ yields a matrix with co-rank 1. Furthermore, removing one additional column corresponding to an edge $e \in E$ that leaves the remaining edges of $E \setminus T$ connected yields a full rank matrix.
\end{prop}

\begin{proof}
The basic idea is to show that removing these edges does not change the span of the remaining vectors. For the spanning tree $T$, we can use the following leaf clipping argument. For an edge adjacent to a leaf of a tree, no other edges in the mesh adjacent to the leaf can be in the tree. Thus, as the sum of the adjacent edge vectors is 0, the vector corresponding to the tree edge is a linear combination of the remaining edges. We can then proceed inductively. Since we know the rank of $C$ is $|E| - |V|$, removing these $|T| = |V| - 1$ edges results in a co-rank 1 matrix.

For $T \cup \{e\}$, we can again apply the same leaf clipping argument, but in the base case we clip our tree to a single triangle. For this case, we have three edge vectors left and three independent equations corresponding to the vertices of the triangle, which uniquely determines expressions of these three edge vectors in terms of the other edge vectors in the span. We thereby obtain a matrix of full rank $|E| - |V|$.
\end{proof}

\subsection{Proof of Continuity for maps between Hyperbolic Surfaces}

\begin{prop}
Suppose two hyperbolic surfaces $H(M, \ell)$ and $H(M, \ell')$ are defined with respect to the same connectivity, and suppose $\tau \in \bR^{|H|}$ satisfies the linear constraints
\[
\begin{aligned}
Z\tau & = 0 \\
C\tau & = \shear' - \shear
\end{aligned}
\]
We equip each triangle $T_{ijk}$ in $M$ with an equilateral reference triangle $T^{\text{ref}}$, and we define the map $P^{\tau}_{ijk}: T^{\text{ref}} \to T^{\text{ref}}$ in terms of (unnormalized) barycentric coordinates on the equilateral reference triangle by
\[
P^{\tau}_{ijk}(w_i, w_j, w_k) = 
\begin{bmatrix}
e^{2(\tau_{ki} - \tau_{ij})/3}w_i \\
e^{2(\tau_{ij} - \tau_{jk})/3}w_j \\
e^{2(\tau_{jk} - \tau_{ki})/3}w_k
\end{bmatrix}
\]
We then define the map $P^{\tau}: H(M, \lambda) \to H(M, \lambda')$ per triangle by
\[
P^{\tau}|_{T_{ijk}} = P^{\tau}_{ijk}
\]
$P^{\tau}$ is well-defined and continuous.
\end{prop}

\begin{figure}
    \includegraphics[width=\columnwidth]{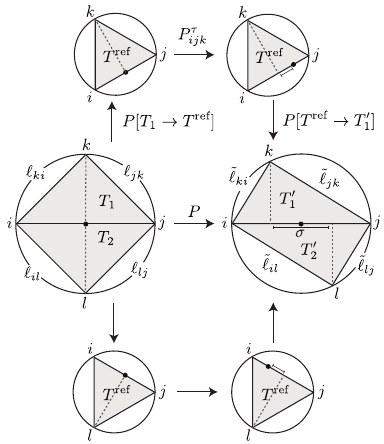}
    \caption{A map between two triangle charts, one with zero shear on a common edge, the other with shear $\shear$.}
    \label{fig:reparametrization}
\end{figure}

\begin{proof}
Let $P[T \to T']$ denote the circumcircle preserving projective maps between triangles $T$ and $T'$ defined in \cite{campen2021efficient}. This map is given in terms of barycentric coordinates by
\[
P[T \to T'](w_i, w_j, w_k) = (S_i w_i, S_j w_j, S_k w_k)
\]
with $S_i = \frac{\ell_{ij} \tilde{\ell}_{jk} \ell_{ki}}{\tilde{\ell}_{ij} \ell_{jk} \tilde{\ell}_{ki}}$, where $\ell$ are the lengths of the edges of the source triangle and $\tilde{\ell}$ are the lengths of the edges of the target triangle.

Let $ij \in E$ with adjacent triangles $T_{ijk}$ and $T_{jil}$, and consider the two 2-triangle charts $(T_1, T_2)$ and $(T_1', T_2')$ corresponding to this edge in $(M, \lambda)$ and $(M, \lambda')$ respectively. The maps between these two 2-triangle charts induced by $P^{\tau}$ is given by
\[
P[T^{\text{ref}} \to T_1'] \circ P^{\tau}_{ijk} \circ P[T_1 \to T^{\text{ref}}]
\]
on $T_1$ and
\[
P[T^{\text{ref}} \to T_2'] \circ P^{\tau}_{jil} \circ P[T_2 \to T^{\text{ref}}]
\]
on $T_2$. Since these maps are projective and thus continuous in the interior of the triangles $T_1, T_2$, it suffices to show that they agree on the common edge $ij$ to establish continuity.

The barycentric coordinate map for $T_1$ restricted to the edge $e_{ij}$ is
\[
\begin{bmatrix}
w_i \\
w_j
\end{bmatrix} \mapsto
\begin{bmatrix}
e^{2(\tau_{ki} - \tau_{ij})/3}
\frac{\ell_{jk}'}{\ell_{ki}' \ell_{ij}'}
\frac{\ell_{ki} \ell_{ij}}{\ell_{jk}} w_i \\
e^{2(\tau_{ij} - \tau_{jk})/3}
\frac{\ell_{ki}'}{\ell_{ij}' \ell_{jk}'} 
\frac{\ell_{ij} \ell_{jk}}{\ell_{ki}} w_j
\end{bmatrix}
\]
where $\ell$ are the lengths of the edges in the two triangle chart $(T_1, T_2)$ and $\ell'$ are the lengths in $(T_1', T_2')$. Similarly, the barycentric coordinate map for $T_2$ along edge $e_{ij}$ is
\[
\begin{bmatrix}
w_i \\
w_j
\end{bmatrix} \mapsto
\begin{bmatrix}
e^{2(\tau_{ji} - \tau_{il})/3}
\frac{\ell_{lj}'}{\ell_{ji}' \ell_{il}'}
\frac{\ell_{ji} \ell_{il}}{\ell_{lj}} w_i \\
e^{2(\tau_{lj} - \tau_{ji})/3}
\frac{\ell_{il}'}{\ell_{lj}' \ell_{ji}'} 
\frac{\ell_{lj} \ell_{ji}}{\ell_{il}} w_j
\end{bmatrix}
\]
To show continuity, we must show that the ratio of the weightings of the barycentric coordinates for the weightings are the same, so the continuity condition is thus
\[
\frac{e^{2(\tau_{ki} - \tau_{ij})/3}}{e^{2(\tau_{ij} - \tau_{jk})/3}}
\frac{(\ell_{jk}')^2}{(\ell_{ki}')^2}
\frac{\ell_{ki}^2}{\ell_{jk}^2} 
=
\frac{e^{2(\tau_{ji} - \tau_{il})/3}}{e^{2(\tau_{lj} - \tau_{ji})/3}}
\frac{(\ell_{lj}')^2}{(\ell_{il}')^2}
\frac{\ell_{il}^2}{\ell_{lj}^2}
\]
Expressed logarithmically with $\lambda = 2 \ln \ell$, the continuity condition is
\[
\begin{aligned}
\frac{2}{3}(\tau_{ki} & - \tau_{ij} - \tau_{ij} + \tau_{jk})
 + \lambda_{jk}' - \lambda_{ki}'
+ \lambda_{ki} - \lambda_{jk}\\
&=
\frac{2}{3}(\tau_{ji} - \tau_{il} - \tau_{lj} + \tau_{ji})
+ \lambda_{lj}' - \lambda_{il}' + 
\lambda_{il} - \lambda_{lj}
\end{aligned}
\]
Since $\tau_{ij} + \tau_{jk} + \tau_{ki} = 0$, we have
\[
\tau_{ij} + \tau_{ji} =
\frac{1}{2}(\lambda_{jk}' - \lambda_{ki}' + \lambda_{il}' - \lambda_{lj}')
- \frac{1}{2}(\lambda_{jk} - \lambda_{ki} + \lambda_{il} - \lambda_{lj}) = \shear_{ij}' - \shear_{ij}
\]
which follows from the per edge constraints, and so the continuity condition is satisfied.
\end{proof}

\subsection{Projective maps from hyperbolic translations}
\begin{prop} Given an ideal hyperbolic triangle $T_{ijk}$ and signed hyperbolic distances $\tau_{ij}, \tau_{jk}, \tau_{ki}$ per edge such that 
\[
\tau_{ij} + \tau_{jk} + \tau_{ki} = 0
\]
there is a unique projective map $P:T_{ijk} \to T_{ijk}$ such that
\begin{align*}
d_H(x, P(x)) = \tau_{ij} && \forall x \in e_{ij} \\
d_H(y, P(y)) = \tau_{jk} && \forall y \in e_{jk} \\
d_H(z, P(z)) = \tau_{ki} && \forall z \in e_{ki} \\
\end{align*}
where $d_H$ is the (signed) hyperbolic distance in the Beltrami-Klein model. Conversely, any projective map $P:T_{ijk} \to T_{ijk}$ satisfies the above property for some values of $\tau_{ij}, \tau_{jk}, \tau_{ki}$ that sum to zero.
\end{prop}

\begin{proof}
Without loss of generality, we use the reference triangle $T_{ijk} \subset \mathbb{R}^3$ with vertices at the standard basis elements so that the normalized barycentric coordinates are the same as the Euclidean coordinates. For a point $x$ with barycentric coordinates $(w_i,w_j,w_k) \in T_{ijk}$, we define $P(x)$ by the (unnormalized) barycentric coordinates
\[
\begin{bmatrix}
e^{2(\tau_{ki} - \tau_{ij})/3}w_i \\
e^{2(\tau_{ij} - \tau_{jk})/3}w_j \\
e^{2(\tau_{jk} - \tau_{ki})/3}w_k
\end{bmatrix}
\]
Since $e^t > 0$ for all $t \in \bR$, this map is well defined, and it clearly maps each vertex to itself.  Furthermore, we have for a point $x \in e_{ij}$ with barycentric coordinates $(w_i, w_j, 0)$ that the hyperbolic distance $d_H$ between $x$ and $P(x)$ is given by
\[
\begin{aligned}
d_H(x, P(x)) & = \frac{1}{2}\ln\left(\frac{(e^{2(\tau_{ij} -
\tau_{jk})/3} w_j l_{ij})(w_i l_{ij})}{(w_j l_{ij})(e^{2(\tau_{ki} - \tau_{ij})/3}w_i l_{ij})}\right) \\
& = \frac{1}{2}\ln\left(\exp\left(\frac{2}{3}(\tau_{ij} - \tau_{jk}) - \frac{2}{3}(\tau_{ki} - \tau_{ij})\right)\right) \\
& = \frac{1}{3}(2\tau_{ij} - \tau_{jk} - \tau_{ki}) \\
& = \tau_{ij} 
\end{aligned}
\]
where $l_{ij}$ is the Euclidean length of $e_{ij}$. By similar direct calculations, we also have the result for the edges $e_{jk}$ and $e_{ki}$.

Uniqueness follows from the fact that there is a unique projective map that sends four points that are not colinear to four points that are not colinear. $P$ must send the three vertices $v_i, v_j, v_k$ to themselves; we show that the image of the midpoint $m$ of $T_{ijk}$ is determined by $\tau$. Consider the lines $l_i, l_j, l_k$ between $v_i, v_j, v_k$ and $m$ respectively. The intersection of these lines is $m$, and they intersect the opposite edges of the triangles at the midpoints $m_{jk}, m_{ki}, m_{ij}$. As $P$ is projective, it maps these lines to lines $P(l_i), P(l_j), P(l_k)$. Moreover, the position of $P(m_{ij})$ is determined by the condition $d_H(m_{ij}, P(m_{ij})) = \tau_{ij}$, and likewise for $m_{jk}, m_{ki}$. Thus, $P(l_i), P(l_j), P(l_k)$ are determined by the translations, so their intersection $P(m)$ is also determined by the translations. As lines that intersect $e_{jk}$ and $v_i$ are in one to one correspondence with $\tau_{ij}$ and likewise for $\tau_{jk}$, distinct translations correspond to distinct projective maps.

Conversely, for a general projective map $P:T_{ijk} \to T_{ijk}$, we have that $P(m)$ must lie in the interior of $T_{ijk}$, and there is some choice of $\tau$ such that the projective map they induce maps $m$ to $P(m)$, so by uniqueness $P$ must be this projective map.
\end{proof}

\end{document}